\documentclass[pra,twocolumn,floatfix,superscriptaddress,reprint,notitlepage,longbibliography]{revtex4-1}

\usepackage[T1]{fontenc}
\usepackage[utf8]{inputenc}
\usepackage{amsmath,amssymb,amsthm,mathtools,mathrsfs,bm}
\usepackage{enumitem}
\usepackage{multirow}
\usepackage{nicefrac}
\usepackage{xspace}
\usepackage{centernot}
\usepackage[colorlinks=true,citecolor=magenta,linkcolor=magenta,urlcolor=magenta]{hyperref}

\usepackage{graphicx,xcolor}
\usepackage{float}
\usepackage{ragged2e}
\usepackage{caption}
\usepackage{subcaption}

\DeclareCaptionJustification{fulljustified}{\justifying}

\def\bra#1{{\left\langle #1 \right|}}
\def\ket#1{{\left| #1 \right\rangle}}
\def\sbra#1{{\langle\!\langle #1 |}}
\def\sket#1{{| #1 \rangle\!\rangle}}
\newcommand{\ketbra}[2]{| #1 \rangle \langle #2 |}
\newcommand{\norm}[1]{\left\lVert #1 \right\rVert}
\newcommand{\Tr}{\operatorname{Tr}}
\newcommand{\tr}{\operatorname{tr}}
\newcommand{\Span}{\operatorname{Span}}
\usepackage{txfonts}

\newcommand{\id}{\mathbb{I}}
\newcommand{\CRB}{Clifford RB\xspace}
\newcommand{\DRB}{direct RB\xspace}
\newcommand{\FOGI}{FOGI\xspace}
\newcommand{\Fvec}{\bm{F}}
\newcommand{\Qmat}{\bm{Q}}
\newcommand{\QHmat}{\bm{Q}_H}
\newcommand{\Lmat}{\mathscr{L}}

\newtheorem{theorem}{Theorem}
\newtheorem{lemma}{Lemma}
\newtheorem{proposition}{Proposition}

\begin{document}

\title{The Impact of Markovian Errors on Random Circuits}
\author{Anthony M. Polloreno}
\email{ampolloreno@gmail.com}
\affiliation{JILA and University of Colorado, Boulder, Colorado 80309, USA}
\affiliation{Quantum Performance Laboratory, Sandia National Laboratories, Livermore, California 94550, USA}
\author{Timothy Proctor}
\affiliation{Quantum Performance Laboratory, Sandia National Laboratories, Livermore, California 94550, USA}
\author{Robin Blume-Kohout}
\affiliation{Quantum Performance Laboratory, Sandia National Laboratories, Livermore, California 94550, USA}

\begin{abstract}
Randomized benchmarking (RB) returns a single decay parameter or gate error rate, but when the noise is gate-dependent that parameter is notoriously difficult to interpret. We study RB and related random-circuit benchmarks under small gate-local Markovian noise by factorizing each noisy gate into its ideal action followed by the exponential of a small error generator and then perturbing the RB $\Lmat$ matrix, which models the RB decay. For any low-error gate set that satisfies the requirements of RB (it is a sequence-asymptotic unitary 2-design), we show that Pauli-stochastic error rates contribute linearly to the RB decay parameter, while Hamiltonian error rates contribute through a gauge-invariant quadratic form. Pauli-correlation and active generators are absent from the linear term and enter only at higher order. For the single-qubit native gate set $\{X_{\pi/2},Y_{\pi/2}\}$ we derive this quadratic form explicitly in terms of three first-order gauge-invariant (\FOGI) coordinates: the two intrinsic over-rotations and one relational axis-tilt error. The resulting formula explains the dependence of \DRB on the sampling distribution and the compilation dependence of \CRB. These results turn RB decay rates into interpretable functions of experimentally meaningful elementary error rates.
\end{abstract}

\maketitle

\section{Introduction}\label{sec:intro}
Randomized benchmarking (RB) \cite{Knill2008,Magesan2011,Magesan2012,Hashim2025} and its descendants are among the standard tools for characterizing quantum processors. In its most common form, RB compresses the behavior of a gate set into a single decay parameter, or equivalently an error rate $r$, extracted from families of random circuits. That single number is routinely used to drive closed-loop gate calibration and tune-up \cite{Kelly2014,Rol2017,WallmanCoherence2015,Sheldon2016}, to cross-check tomographic and model-based characterizations \cite{Merkel2013,BlumeKohout2017,Nielsen2021,Madzik2022}, and to summarize hardware performance for larger-scale circuit diagnostics \cite{Erhard2019,Proctor2021,Hashim2025}.

Over the last decade, RB has been generalized in many directions. Representative examples include interleaved, simultaneous, correlated, loss, leakage, and coherence-sensitive variants of \CRB \cite{MagesanInterleaved2012,Gambetta2012Simultaneous,McKay2020,WallmanLoss2015,Chasseur2015,Wood2018,WallmanCoherence2015,Feng2016,Sheldon2016}, benchmarking protocols for non-Clifford or otherwise restricted gate sets \cite{CarignanDugas2015,Cross2016,BrownEastin2018,Hashagen2018,Helsen2019,Claes2021,Helsen2022General,Helsen2022Matchgate,FrancaHashagen2018}, and scalable native-gate protocols such as direct, mirror, and binary RB \cite{Proctor2019,Polloreno2025,Hines2023,Hines2024,Mayer2021,Heinrich2022}. Statistical, drift, and non-Markovian aspects of RB have also been studied extensively \cite{WallmanFlammia2014,Epstein2014,Granade2015,Hincks2018,Harper2019,Fogarty2015,Mavadia2018,Proctor2020}.

Despite this success, interpreting the reported RB error rate remains subtle. In the gate-independent noise model analyzed in the original RB theory \cite{Magesan2011,Magesan2012}, every implemented gate is the ideal gate followed by the same error channel $\mathcal E$, and the RB error rate is simply the infidelity of $\mathcal E$ \cite{Nielsen2002}. Real devices, however, exhibit gate-dependent noise. In that setting the naive identification of the RB error rate with an average gate infidelity fails \cite{Proctor2017,Wallman2018}. The difficulty is fundamental; the average infidelity of a gate set is gauge dependent, so it cannot be what RB measures \cite{Proctor2017,Nielsen2021,Madzik2022,GonzalezDeMendoza2026}.

A theory already exists for RB with gate-dependent noise. Exact or asymptotically exact formulas can be written using the $\Lmat$ matrix (reviewed later in this paper), Fourier analysis on groups, or closely related superchannel constructions \cite{Proctor2017,Wallman2018,Merkel2021,Polloreno2025}. Related work connects RB decay parameters to gate-set circuit fidelity and other gauge-invariant figures of merit \cite{CarignanDugas2018,Boone2019}. Those results showed that, under broad conditions, RB is reliable---its error rate is both well-defined, because RB success probabilities decay exponentially, and also well-motivated. However, these theories do not explain how specific errors, e.g., a coherent over-rotation error, impact $r$.

The goal of this paper is to supply such an interpretation, in the perturbative regime of small gate-local Markovian noise. We factor each noisy gate into its ideal action followed by the exponential of a small error generator expressed in the Lindbladian basis of Hamiltonian, Pauli-stochastic, Pauli-correlation, and active sectors \cite{BlumeKohout2022}. We use the FOGI framework of Gonzalez de Mendoza \emph{et al.} \cite{GonzalezDeMendoza2026} to identify the gauge-invariant coordinate layer, and then study the RB $\Lmat$ matrix for sequence-asymptotic unitary 2-designs (gate sets for which the second moment of an i.i.d.\ gate sequence converges to the Haar second moment), a setting that encompasses standard Clifford RB, direct RB, and closely related random-circuit benchmarks \cite{Proctor2019,Polloreno2025,Hines2023,Hines2024}. Conceptually, this paper adds a perturbative interpretation layer on top of the $\Lmat$-matrix and sequence-asymptotic-unitary-2-design framework developed in Ref.~\cite{Polloreno2025}. This viewpoint is closely connected to twirling and unitary-design constructions \cite{Emerson2005,Emerson2007,Gross2007,Dankert2009} and to the broader random-circuit literature on design growth and operator spreading \cite{BrownFawzi2015,Brandao2016,Oszmaniec2022,HunterJones2018,Keyserlingk2018,Nahum2018}.

Our main results are as follows. First, for every fixed sampling distribution with the required second-moment gap, we show that, in a distribution-dependent neighborhood of the ideal gate set, Pauli-stochastic error generators are the only generators that contribute linearly to the RB decay parameter. Second, Hamiltonian generators contribute through a gauge-invariant quadratic form, so the coherent contribution to the RB error rate depends only on gauge-invariant combinations of Hamiltonian error rates and can be expressed in FOGI coordinates \cite{Madzik2022,GonzalezDeMendoza2026}. Third, for the native one-qubit gate set $\{X_{\pi/2},Y_{\pi/2}\}$ we derive this quadratic form explicitly in two intrinsic over-rotation coordinates and one pairwise relational axis-tilt coordinate. These formulas explain both the sampling dependence of direct RB and the compilation dependence of Clifford RB. Our results complement noise-tailoring approaches such as Pauli-frame randomization and randomized compiling \cite{Knill2005,Wallman2016,Ware2021,Hashim2021, polloreno2022robustly}. Instead of averaging coherent errors away, we identify how their gauge-invariant components contribute to the RB decay rate.

\begin{figure}[t!]
\includegraphics[width=0.98\linewidth]{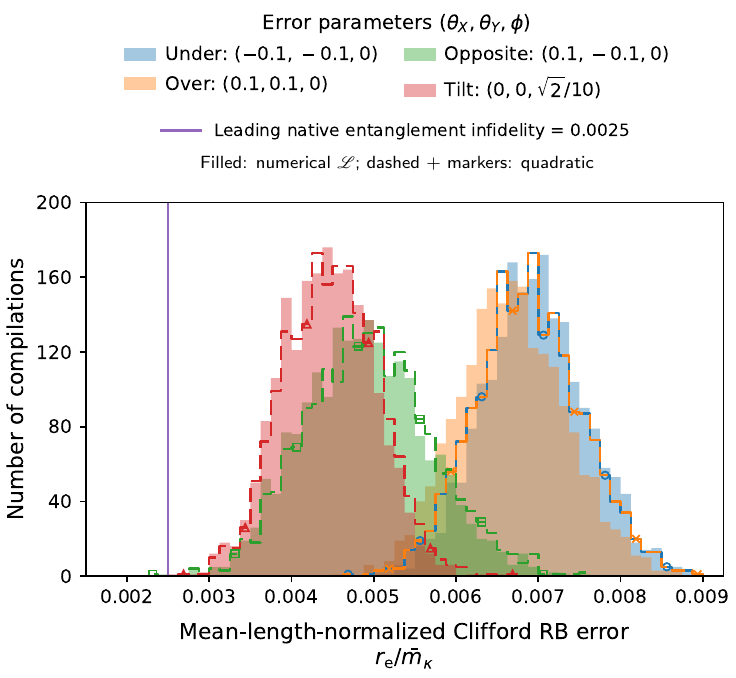}
\centering
\caption{How one-qubit \CRB error rates depend on compilation for four coherent native-gate error models. Filled histograms use the exact compiled-Clifford $\Lmat$ matrix, while dashed step curves and markers use the quadratic formula derived in Sec.~\ref{sec:crb} and Appendix~\ref{app:crb-manual}. Each distribution shows $r_{\mathrm e}/\bar m_\kappa$ for 2000 random compilations of the 24 single-qubit Cliffords into native $X_{\pi/2}$ and $Y_{\pi/2}$ sequences of length at most five, where $\bar m_\kappa$ is the compilation's mean native-gate count. We normalize $r_{\mathrm e}$ by $\bar m_\kappa$ because it is useful for comparison; this does not imply that the normalized value is a formal native-gate error rate. The vertical line is the common leading native-gate entanglement infidelity, $\theta^2/4=\phi^2/8=0.0025$. The equal-over-rotation and equal-under-rotation models have the same quadratic prediction; their markers are offset so that both estimates are visible. Their exact distributions differ because of higher-order terms.}
\label{fig:crb_example}
\end{figure}

\section{Example: coherent gate errors in RB}\label{sec:motivation}

The problem we want to solve can be illustrated by two simple examples involving one-qubit gates with purely coherent errors. Suppose an experimentalist performs RB with a set of two single-qubit gates, $X_{\nicefrac{\pi}{2}}$ and $Y_{\nicefrac{\pi}{2}}$, whose ideal action is given by the $2\times 2$ unitaries
\begin{align}
    X_{\nicefrac{\pi}{2}} &= \exp\left(-\nicefrac{i\pi}{4} X \right) = \cos\left(\nicefrac{\pi}{4}\right)\id - i\sin\left(\nicefrac{\pi}{4}\right)X, \label{eq:Xideal}\\
    Y_{\nicefrac{\pi}{2}} &= \exp\left(-\nicefrac{i\pi}{4} Y \right) = \cos\left(\nicefrac{\pi}{4}\right)\id - i\sin\left(\nicefrac{\pi}{4}\right)Y, \label{eq:Yideal}
\end{align}
but which \textit{actually} act as
\begin{align}
    \widetilde{X}_{\nicefrac{\pi}{2}}(\theta_X,\phi) = \exp[\nonumber
    -\nicefrac{i}{2}&\left(\nicefrac{\pi}{2} + \theta_X\right) \times\\&\left(\cos\left(\nicefrac{\phi}{2}\right)X + \sin\left(\nicefrac{\phi}{2}\right)Y\right)], \label{eq:Xmodel}\\
    \widetilde{Y}_{\nicefrac{\pi}{2}}(\theta_Y,\phi) = \exp[-\nicefrac{i}{2}&\left(\nicefrac{\pi}{2} + \theta_Y\right)\times\nonumber
    \\&\left(\cos\left(\nicefrac{\phi}{2}\right)Y + \sin\left(\nicefrac{\phi}{2}\right)X\right)]. \label{eq:Ymodel}
\end{align}
This error model incorporates two kinds of coherent gate error: (1) over-rotation, parameterized by $\theta_X$ and $\theta_Y$, and (2) a tilt error, parameterized by $\phi$, that changes the angle between the two gates' rotation axes. When $\theta_X=\theta_Y=\phi=0$, both gates are error free, and we expect RB experiments to yield $r=0$.

The behavior of $r$ for these gates can be explored either by exact calculations from the $\Lmat$ matrix (Sec.~\ref{sec:error-factorization}) or, in the small-error regime, by the quadratic formulas derived in Secs.~\ref{sec:drb-xy} and \ref{sec:crb}. Doing so shows (Figs.~\ref{fig:crb_example} and \ref{fig:drb_example}) that $r$ depends in nontrivial and arguably surprising ways both on the parameters $\theta_X$, $\theta_Y$, and $\phi$ and on how the random circuits are constructed.

\subsection{Clifford RB}

Clifford RB \cite{Magesan2011,Magesan2012,Hashim2025} is the most common and standard form of RB, but it does not directly benchmark native gates. To perform \CRB, we must construct, or \textit{compile}, each of the 24 single-qubit Clifford operations using the available gates $X_{\nicefrac{\pi}{2}}$ and $Y_{\nicefrac{\pi}{2}}$. There is no unique way to compile the Clifford gates into native $X_{\nicefrac{\pi}{2}}$ and $Y_{\nicefrac{\pi}{2}}$ gates---infinitely many compilations are possible, and even focusing on sensible ones does not identify a unique best option. Experimenting with different compilations, we find that both the value of $r$ and its dependence on the individual gates' errors depend nontrivially on the compilation; see also the related observation by Boone \emph{et al.} \cite{Boone2019}.

We varied both the error model and the compilation of the 24 Clifford gates into native gates. In Figure~\ref{fig:crb_example}, we show results for 2000 sampled Clifford compilations. They were generated by enumerating all implementations of each one-qubit Clifford gate using sequences of at most five native gates, then randomly selecting one implementation for each Clifford gate. The four representative error models are
\begin{enumerate}[label=(\arabic*)]
    \item equal over-rotation errors on both gates ($\theta_X=0.1$, $\theta_Y=0.1$, and $\phi=0$),
    \item equal under-rotation errors on both gates ($\theta_X=-0.1$, $\theta_Y=-0.1$, and $\phi=0$),
    \item one gate under-rotates and the other over-rotates by the same amount ($\theta_X=0.1$, $\theta_Y=-0.1$, and $\phi=0$),
    \item an error in the angle between the two gates' rotation axes, with $\theta_X=0$, $\theta_Y=0$, and $\phi=\sqrt{2}\times 0.1$, chosen so that the native-gate infidelity matches the previous three models at leading order.
\end{enumerate}
For each compilation and error model, we constructed the exact noisy Clifford maps by multiplying the full noisy native-gate maps defined by Eqs.~\eqref{eq:Xmodel} and \eqref{eq:Ymodel} in circuit order, without expanding in the error parameters. We then numerically diagonalized the resulting $9\times9$ traceless-sector $\Lmat$ matrix (see Sec.~\ref{sec:error-factorization}) to obtain the eigenvalue $\gamma$ connected continuously to the ideal depolarizing mode and the rate $r_{\mathrm e}=\frac34(1-\gamma)$. The quadratic predictions were computed separately using Eq.~\eqref{eq:crb-quadratic-form} with the first-order compiled Hamiltonian vectors of Eq.~\eqref{eq:compiled-hC}. Since the average number of native gates per Clifford varies between compilations, we plot $r_{\mathrm e}/\bar m_\kappa$. This is a reporting normalization for comparing compilations, not a formal native-gate error rate.

Figure~\ref{fig:crb_example} shows that the normalized rates vary significantly across compilations. Although all four models have the same leading-order native-gate entanglement infidelity, the opposite-sign over/under-rotation model and the pure-tilt model generate significantly different histograms, while the equal-over/equal-under quadratic estimates coincide and their exact distributions differ at higher order. Unambiguously, the \CRB error rate is not a weighted average of native-gate infidelities. It is sensitive to finer, gauge-invariant structure in the coherent error model, and that sensitivity is filtered through the compilation. We are therefore motivated to ask what property of the gates Clifford RB is measuring. More precisely, how does $r$ vary with $\theta_X$, $\theta_Y$, $\phi$, and the choice of compilation?

\begin{figure}[t!]
\includegraphics[width=0.98\linewidth]{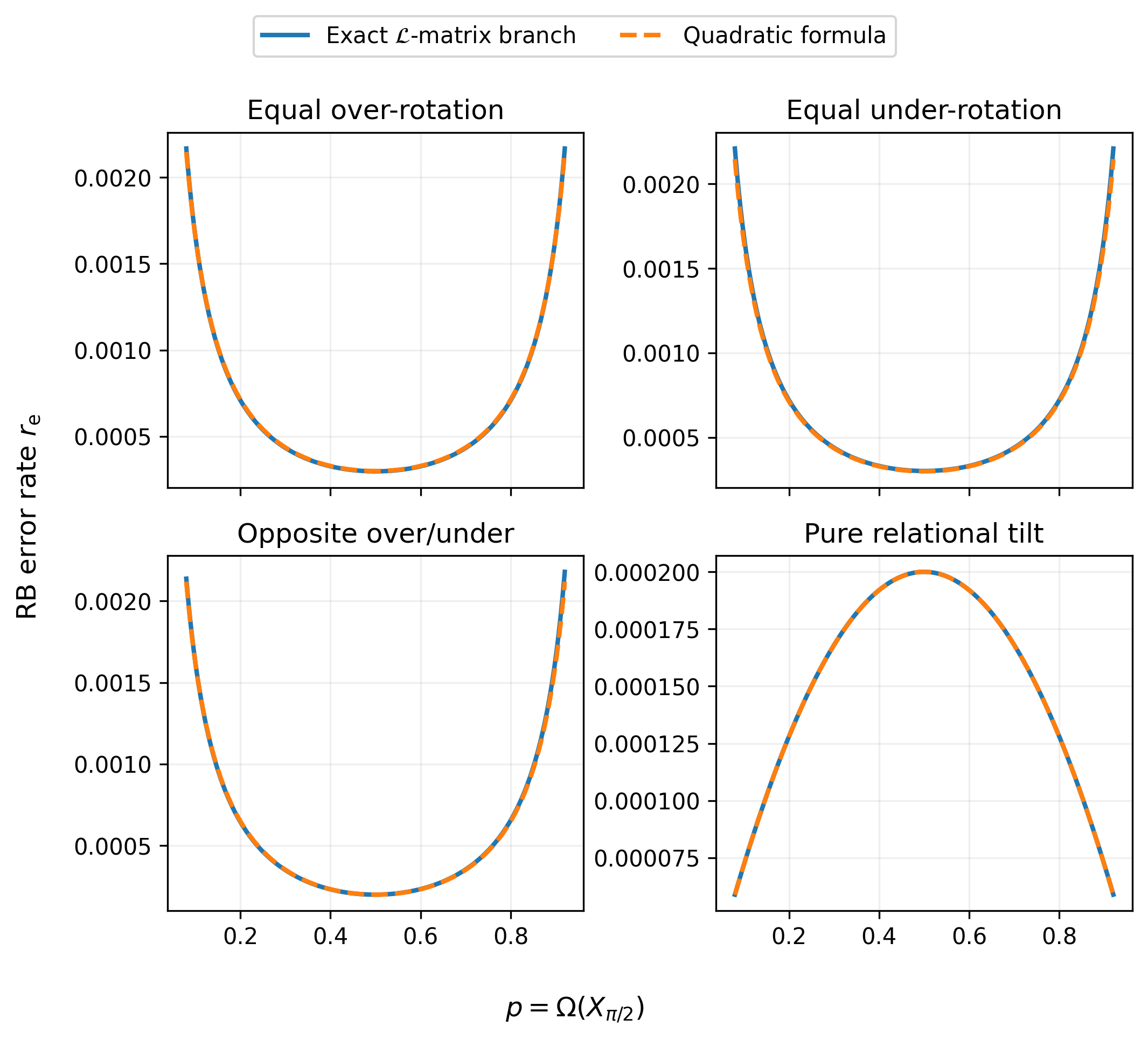}
\caption{Direct RB error rate $r_{\mathrm e}$ versus $p=\Omega_X$ for the same four coherent error models shown in Fig.~\ref{fig:crb_example}, but with all amplitudes reduced by a factor of five: $|\theta_X|=|\theta_Y|=0.02$ for the first three models and $\phi=\sqrt{2}\times0.02$ for pure tilt. These plots confirm that the \DRB error rate depends non-trivially on the sampling parameter $p$, as discussed in the main text. Parameter magnitudes were chosen to keep the model within the perturbative regime. Solid curves show the \DRB error rate predicted by the most accurate theory for \DRB \cite{Polloreno2025}, and correspond to the exact $\Lmat$-matrix eigenvalue branch connected to the ideal depolarizing mode. Dashed curves are the quadratic formulas of Sec.~\ref{sec:drb-xy}. On the displayed interval $0.08\leq p\leq0.92$, this branch is the unique dominant nontrivial mode for all four models, and any difference between the two predictions is due to higher order terms. The displayed interval is restricted because the perturbative approximation is not uniform in \(p\): at the fixed error amplitude used here, its quantitative accuracy and the isolation of the depolarizing decay mode deteriorate near the endpoints. As the error amplitude tends to zero, the interval of agreement expands toward the full open interval \(0<p<1\).}
\label{fig:drb_example}
\end{figure}

\subsection{Direct RB}

We began by examining Clifford RB because it is ubiquitous. However, the behavior of Clifford RB necessarily convolves the effect of the errors with the effect of compilation. To understand whether all the effects seen in Figure~\ref{fig:crb_example} arise from compilation, we turn instead to a simpler protocol with no compilation step: direct RB.

Direct RB benchmarks native gate sets directly by running circuits comprised almost entirely of a sequence of operations drawn independently from the native gate set \cite{Proctor2019,Polloreno2025}. The direct RB circuits also contain short subroutines at the beginning and end---essentially a random Clifford and the inverse of the full sequence---but those pieces are not important for our purposes here. The gates in the main body of the circuit may, but need not, be sampled from a uniform distribution over the gate set. Instead, direct RB is parameterized by a distribution $\Omega$. For our example, $\Omega=(\Omega_X,\Omega_Y)$, where $\Omega_X$ and $\Omega_Y$ are the probabilities of sampling $X_{\nicefrac{\pi}{2}}$ and $Y_{\nicefrac{\pi}{2}}$, respectively, and sum to 1.

Understanding the behavior of $r$ in direct RB is simpler than in Clifford RB because there is no Clifford-compilation step. For this gate set, the only free parameter is $\Omega_X=1-\Omega_Y$. A natural heuristic would be $r \approx \Omega_X\epsilon_X+\Omega_Y\epsilon_Y$, where $\epsilon_X$ and $\epsilon_Y$ are the entanglement infidelities of the two gates. This is not entirely naive---for purely stochastic errors, it is valid to leading order, and higher-order corrections are tiny in most circumstances \cite{Polloreno2025}---but it cannot be precisely correct in general because infidelity is not gauge invariant.

Figure~\ref{fig:drb_example} shows how badly this heuristic fails for coherent errors.

We computed \DRB's $r$ as a function of $\Omega_X$ along each of the four error directions defined above. To display the accuracy of the quadratic theory, Fig.~\ref{fig:drb_example} reduces every error amplitude by a factor of five relative to Fig.~\ref{fig:crb_example}. Since the two gates still have the same leading native-gate infidelity in each case, the naive weighted-infidelity heuristic suggests that all four curves should be flat (i.e., straight lines with no gradient). Instead, all four curves depend strongly on $\Omega_X$, and they are neither flat, nor linear, nor the same across the four error models.

A full explanation of these curves must wait for the theory developed in the rest of the paper, but a heuristic picture is already helpful. The first three error models are dominated by over-rotation errors that are \textit{intrinsic} to each gate in the FOGI classification of Ref.~\cite{GonzalezDeMendoza2026}. When the same gate appears several times consecutively, those errors can interfere constructively and produce a large $r$. This tends to happen when $\Omega_X$ is close to 0 or 1. When $\Omega_X\approx \nicefrac{1}{2}$, by contrast, the two gates alternate more often and effectively twirl each other's intrinsic errors before they can add up. The fourth model is the corresponding pairwise \textit{relational} tilt error. Each gate's tilted rotation axis can be aligned with its ideal axis by a change of reference frame, but no single frame simultaneously aligns both erroneous axes because their relative angle has changed. The quadratic prediction therefore tends to zero as $\Omega_X$ approaches 0 or 1 and is maximized near $\Omega_X=\nicefrac{1}{2}$. These same effects occur in \CRB, but they are more complex to describe because they depend on the details of the compilation.

\subsection{Discussion}

These simple examples demonstrate that interpreting, or even predicting, $r$ is nontrivial for both Clifford RB and direct RB. In the rest of the paper, we develop a theory for predicting and interpreting $r$. The key observation is that, in the limit of small coherent errors, and for any fixed compilation in Clifford RB or fixed sampling distribution in direct RB, $r$ is a quadratic function of $\theta_X$, $\theta_Y$, and $\phi$. This simplifies the problem greatly: we need only identify the quadratic form and then understand what its components mean.

Our first concern will therefore be to prove that $r$ is quadratic in the coherent error rates. Next, we derive an expression for the particular quadratic form of $r$. Once this is done, we turn to the more general questions of how stochastic and coherent errors combine, and how the same structure appears across different random-circuit benchmarks and different gate sets.

\section{Error models}\label{sec:error-models}

We begin by introducing the error generator formalism and gauge structure used throughout the paper.

\subsection{The small Lindbladian error basis}

In this work, we study the impact of Markovian errors on random circuits containing $n$-qubit gates from some set $\{\mathcal{G}_k\}$. Here $\mathcal{G}_k$ denotes the superoperator for an error-free unitary gate, and $\widetilde{\mathcal{G}}_k$ denotes the implemented, imperfect gate. Our gate-local Markovian model assigns each gate label a fixed completely positive and trace-preserving (CPTP) map on the modeled $n$ qubits. That map is independent of circuit history, and a circuit channel is the composition of its gate maps in circuit order. This excludes drift, context dependence, temporal correlations, and unmodeled leakage. We are concerned with how different errors affect random circuits, and RB in particular, so we need a parameterized model for the errors in a gate. For this we use the small Lindbladian error basis \cite{BlumeKohout2022}.

There are four kinds of errors in this basis: Hamiltonian (H), stochastic (S), correlation (C), and active (A). They are defined by
\begin{align}
    H_P[\rho] &= -i[P,\rho], \\
    S_P[\rho] &= P\rho P - \rho, \\
    C_{P,Q}[\rho] &= P\rho Q + Q\rho P - \frac{1}{2}\bigl\{\{P,Q\},\rho\bigr\}, \\
    A_{P,Q}[\rho] &= i\left(P\rho Q-Q\rho P + \frac{1}{2}\bigl\{[P,Q],\rho\bigr\}\right).
\end{align}
Here $\rho$ is an $n$-qubit density operator, $P,Q\in\mathbb{P}^*$, where $\mathbb{P}^*$ is the set of $n$-qubit Pauli operators without signs and excluding the identity, and $P<Q$ denotes a fixed ordering of Pauli labels.

We can express the implemented gate as the ideal gate followed by a post-gate error process,
\begin{equation}
    \widetilde{\mathcal{G}}_k = e^{L_k}\,\mathcal{G}_k,
    \label{eq:postfactor}
\end{equation}
where the generator of the error process is
\begin{align}
L_k &= \sum_P h_{P}^{(k)} H_P + \sum_P s_{P}^{(k)} S_P \notag\\
&\quad + \sum_{P<Q} c_{P,Q}^{(k)} C_{P,Q} + \sum_{P<Q} a_{P,Q}^{(k)} A_{P,Q}.
\label{eq:generator-expansion}
\end{align}
The coefficients $h_P^{(k)}$, $s_P^{(k)}$, $c_{P,Q}^{(k)}$, and $a_{P,Q}^{(k)}$ are the rates of the four kinds of elementary error on gate $k$. This is a local coordinate expansion: complete positivity constrains combinations of coefficients, and we do not assume that every error map is embeddable in a time-homogeneous Lindblad semigroup. We may therefore parameterize an imperfect gate by a vector $\vec{\epsilon}^{\,(k)}$ of error rates, and a gate set by stacking those vectors into a single parameter vector $\vec{\epsilon}$.

In this paper, we study how the RB error rate $r$ of a gate set depends on $\vec{\epsilon}$. Because each kind of error has a different effect, it is often useful to split the error parameters into the four sectors, with vectors $\vec{h}$, $\vec{s}$, $\vec{c}$, and $\vec{a}$ of Hamiltonian, stochastic, correlation, and active rates, respectively.

\subsection{Linearized gauge transformations and \FOGI coordinates}

Gate set descriptions have a similarity-transformation gauge freedom
\cite{Proctor2017,Nielsen2021,Madzik2022,GonzalezDeMendoza2026}.
Let $\rho_i$ denote a prepared density operator and $E_j$ a measurement
effect associated with outcome $j$. In Liouville notation,
$|\rho_i\rangle\rangle$ represents the prepared state and
$\langle\langle E_j|$ represents the linear functional
$X\mapsto\operatorname{tr}(E_jX)$. The probability of outcome $j$ after
a gate sequence is
\begin{equation}
    p(j\mid i,k_1,\ldots,k_m)
    =\langle\langle E_j|
      \widetilde{\mathcal G}_{k_m}\cdots
      \widetilde{\mathcal G}_{k_1}
      |\rho_i\rangle\rangle.
\end{equation}
If the preparations, measurement effects, and gates are transformed using
the same invertible map $\mathcal M$ as
\begin{equation}
\begin{aligned}
    \sket{\rho_i}
    &\to \mathcal M\sket{\rho_i},\\
    \sbra{E_j}
    &\to \sbra{E_j}\mathcal M^{-1},\\
    \widetilde{\mathcal G}_k
    &\to \mathcal M\widetilde{\mathcal G}_k\mathcal M^{-1},
\end{aligned}
\label{eq:gauge-action}
\end{equation}
then all observable probabilities are unchanged. For small errors, it is
therefore natural to ask which directions in error-generator space are
observable and which are pure gauge.

For the post-gate factorization of Eq.~\eqref{eq:postfactor}, a small
gauge transformation $\mathcal M=e^K$ induces the first-order shift
\begin{equation}
    L_k \longmapsto L_k'
    =L_k+K-\mathcal G_kK\mathcal G_k^{-1}
    +O\!\left(\|K\|\,\|L_k\|+\|K\|^2\right).
    \label{eq:lin-gauge}
\end{equation}
The exact transformed post-gate error map obeys
\[
    e^{L_k'}
    =e^K e^{L_k}e^{-\mathcal G_kK\mathcal G_k^{-1}}.
\]
Thus $K=0$ gives $L_k'=L_k$ exactly. For fixed small $L_k$, the remainder in Eq.~\eqref{eq:lin-gauge} tends to zero as $K\to0$.

A first-order gauge-invariant (FOGI) property is a linear combination of error-generator coefficients that is unchanged by every first-order gauge shift \cite{GonzalezDeMendoza2026}. Let $\mathcal H$ be the stacked Hamiltonian-rate space and $\Gamma_H$ the subspace of shifts generated by arbitrary Hamiltonian gauge generators $K$, i.e., real linear combinations of the generators $H_P$. Ideal unitary conjugation preserves the Hamiltonian sector. A complete set of Hamiltonian FOGI properties therefore assigns the same coordinates to two Hamiltonian error vectors exactly when their difference lies in $\Gamma_H$. These properties
provide coordinates on $\mathcal H/\Gamma_H$; they do not select a unique representative error vector in each equivalence class.
\subsection{\texorpdfstring{The $\{X_{\pi/2},Y_{\pi/2}\}$ gate set}{The \{Xpi/2,Ypi/2\} gate set}}

For the concrete gate model of Eqs.~\eqref{eq:Xmodel} and \eqref{eq:Ymodel}, linearizing the post-gate error generators gives
\begin{align}
    L_X &= \frac{\theta_X}{2}H_X + \frac{\phi}{4}\left(H_Y+H_Z\right) + O\left(\theta_X^2,\phi^2,\theta_X\phi\right), \label{eq:LX}\\
    L_Y &= \frac{\theta_Y}{2}H_Y + \frac{\phi}{4}\left(H_X-H_Z\right) + O\left(\theta_Y^2,\phi^2,\theta_Y\phi\right). \label{eq:LY}
\end{align}
A convenient set of Hamiltonian \FOGI coordinates, consistent with the intrinsic/relational classification of Ref.~\cite{GonzalezDeMendoza2026}, is
\begin{align}
    F_1 &= h_X^{(X)}, \label{eq:F1}\\
    F_2 &= h_Y^{(Y)}, \label{eq:F2}\\
    F_3 &= \frac{1}{2}\left(h_Y^{(X)} + h_Z^{(X)} + h_X^{(Y)} - h_Z^{(Y)}\right). \label{eq:F3}
\end{align}
Writing $\eta=\max\{|\theta_X|,|\theta_Y|,|\phi|\}$, the physical parameterization of Eqs.~\eqref{eq:Xmodel} and \eqref{eq:Ymodel} gives
\begin{equation}
\begin{aligned}
    F_1 &= \frac{\theta_X}{2}+O(\eta^2),\qquad
    F_2 = \frac{\theta_Y}{2}+O(\eta^2),\\
    F_3 &= \frac{\phi}{2}+O(\eta^2).
    \label{eq:F-map}
\end{aligned}
\end{equation}
Thus $F_1$ and $F_2$ are intrinsic over- or under-rotation coordinates. With the axis-misalignment convention of Ref.~\cite[Eqs.~(157)--(159)]{GonzalezDeMendoza2026}, the first-order axis-misalignment angle is $2F_3=\phi+O(\eta^2)$, so $F_3$ is one half of the pairwise relational tilt error.

\section{Modeling fidelity decay in random circuits}\label{sec:error-factorization}

We work in a normalized Hermitian operator basis, so ideal unitary channels are represented by real orthogonal matrices. For a sequence-asymptotic unitary 2-design, the bulk decay spectrum of the mean survival probability is governed by the $\Lmat$ matrix \cite{Proctor2017,Polloreno2025}
\begin{equation}
    \Lmat = \sum_k \Omega_k\, \mathcal{G}_k \otimes \widetilde{\mathcal{G}}_k.
    \label{eq:Lmatrix}
\end{equation}
Following Ref.~\cite{Polloreno2025}, this means that the second moment of an i.i.d. length-$m$ gate sequence converges to the Haar second moment as $m\to\infty$. 

Let
\begin{equation}
    \Lmat_0(\Omega)
    =\sum_k\Omega_k\,\mathcal G_k\otimes\mathcal G_k
\end{equation}

denote the noise-free second-moment matrix. Its $m$th power is the second-moment matrix of a length-$m$ sequence of independently sampled ideal gates. Under column vectorization, the range of the Haar second-moment projector corresponds to the superoperator subspace
\begin{equation}
    \mathcal V_1=\Span\{\Delta,\mathrm{id}\},
    \label{eq:V1}
\end{equation}
where
\begin{equation}
    \Delta(\rho)=\tr(\rho)\frac{\id}{d},
    \qquad
    \mathrm{id}(\rho)=\rho.
\end{equation}
Every ideal unitary channel fixes this subspace under conjugation. Identifying superoperators with their column vectorizations, let $\rho_2(\Omega)$ be the largest absolute value of an eigenvalue of $\Lmat_0(\Omega)$ restricted to the orthogonal complement of $\mathcal V_1$. Define
\begin{equation}
    \Delta_2(\Omega)=1-\rho_2(\Omega).
    \label{eq:design-gap}
\end{equation}
The design condition $\Delta_2(\Omega)>0$ states that every eigenvalue outside $\mathcal V_1$ has modulus strictly less than one. Consequently, $\Lmat_0(\Omega)^m$ converges to the Haar second-moment projector as $m\to\infty$, and $\mathcal V_1$ is the full ideal eigenvalue-one subspace.

The depolarizing direction is $\mathrm{id}-\Delta$. Trace preservation leaves a separate exact unit eigenvalue under perturbation, although its right eigenvector need not remain $\Delta$ for nonunital noise.

Fix $\Omega$ with $\Delta_2(\Omega)>0$. The ideal depolarizing eigenvalue is simple within the invariant traceless--traceless block of $\Lmat$, so it extends to an analytic eigenvalue branch $\gamma$ for sufficiently small errors. For errors sufficiently close to zero, $\gamma$ has larger modulus than every other nontrivial eigenvalue and governs the leading non-constant contribution to the RB survival probability. The required upper bound on the error strength depends on $\Omega$. In particular, we do not assert that a single positive error threshold works for all sampling distributions as their ideal gap approaches zero.

We use the entanglement-infidelity normalization
\begin{equation}
    r\equiv r_{\mathrm e}=\frac{d^2-1}{d^2}(1-\gamma),\qquad d=2^n.
    \label{eq:rfromgamma}
\end{equation}
The conventional average-gate-infidelity normalization \cite{Nielsen2002} is
\begin{equation}
    r_{\mathrm{avg}}=\frac{d}{d+1}r_{\mathrm e}=\frac{d-1}{d}(1-\gamma).
    \label{eq:ravgfromgamma}
\end{equation}
Thus for one qubit $r_{\mathrm e}=\frac34(1-\gamma)$ whereas $r_{\mathrm{avg}}=\frac12(1-\gamma)$. Every numerical value and figure below uses $r\equiv r_{\mathrm e}$.

Figure~\ref{fig:drb_example} compares the eigenvalue branch connected to the ideal depolarizing mode with the quadratic formulas of Sec.~\ref{sec:drb-xy}. The error amplitudes are reduced relative to Fig.~\ref{fig:crb_example} so that this branch remains the unique dominant nontrivial mode throughout the displayed interval. Equal over- and under-rotation share the same quadratic approximation, so the difference between their exact rates is a higher-order effect. At fixed nonzero error strength, taking $p$ sufficiently close to $0$ or $1$ can produce other nontrivial eigenvalues $\lambda$ of $\Lmat$ with $|\lambda|\geq|\gamma|$. If the corresponding modes have nonzero coefficients in the mean survival probability, their contributions are not asymptotically smaller than the contribution from $\gamma$. The survival probability then need not be dominated by a single exponential with decay parameter $\gamma$.

We now present our main perturbative result, relating $r$ to the rates of elementary error generators. First, we define
\begin{equation}
\delta \coloneqq
\max_{k,P,Q}
\Bigl\{|s_P^{(k)}|,|c_{P,Q}^{(k)}|,|a_{P,Q}^{(k)}|\Bigr\},
\qquad
\theta \coloneqq \max_{k,P}|h_P^{(k)}|.
\label{eq:small-parameters}
\end{equation}
Complete positivity constrains these dissipative coefficients jointly. We retain all three sectors in $\delta$ because a CPTP error map need not have a Lindblad logarithm; exact Lindbladian coefficient bounds are therefore not assumed here \cite[Sec.~IV]{BlumeKohout2022}. Our general perturbative result is the following.

\begin{theorem}\label{thm:main-perturbative}
Fix a gate set and sampling distribution $\Omega$ that is a sequence-asymptotic unitary 2-design. There is an $\Omega$-dependent neighborhood of the ideal gate set such that gate-local Markovian errors in the post-gate form of Eq.~\eqref{eq:postfactor} satisfy
\begin{equation}
    \gamma
    = 1
    - \frac{d^2}{d^2-1}\sum_{k,P}\Omega_k\,s_P^{(k)}
    - Q_H(\bm h;\Omega)
    + O(\delta^2,\delta\theta,\theta^3),
    \label{eq:main-expansion}
\end{equation}
where $Q_H(\bm h;\Omega)=\bm h^{\mathsf T}\QHmat(\Omega)\bm h$ is a quadratic form in the Hamiltonian error coefficients and $\QHmat(\Omega)$ is real symmetric. Every term contains exactly two Hamiltonian coefficients; equivalently,
\begin{equation}
    Q_H(\alpha\bm h;\Omega)=\alpha^2Q_H(\bm h;\Omega).
\end{equation}
This quadratic form is unchanged when a linearized Hamiltonian gauge shift is added to $\bm h$, so it depends only on the equivalence class of $\bm h$ in $\mathcal H/\Gamma_H$. Equivalently, for every Hamiltonian gauge generator $K$, with shift $\bm g_H(K)\in\Gamma_H$,
\begin{equation}
    \QHmat(\Omega)\,\bm g_H(K)=0.
    \label{eq:gauge-kernel-thm}
\end{equation}
The expansion is taken as $\delta,\theta\to0$ with $\Omega$ held fixed. More explicitly, the remainder has magnitude at most $C_\Omega(\delta^2+\delta\theta+\theta^3)$ for sufficiently small $\delta$ and $\theta$. Both $C_\Omega$ and the permitted sizes of $\delta$ and $\theta$ may depend on $\Omega$; no bound independent of $\Omega$ is asserted as $\Delta_2(\Omega)\to 0$.
\end{theorem}

An immediate corollary of Theorem~\ref{thm:main-perturbative} and Eq.~\eqref{eq:rfromgamma} is the corresponding expansion for the reported RB error rate:
\begin{equation}
\begin{aligned}
    r\equiv r_{\mathrm e}
    &=\sum_{k,P}\Omega_k\,s_P^{(k)}
      +\frac{d^2-1}{d^2}\,Q_H(\bm h;\Omega)\\
    &\quad+O(\delta^2+\delta\theta+\theta^3).
\end{aligned}
\label{eq:r-main-expansion}
\end{equation}
The expansion is at fixed $\Omega$, and its remainder constants and permitted error magnitudes may depend on $\Omega$.

The remainder of this section proves the linear stochastic term and then shows how analyticity and gauge invariance force the structure of the quadratic Hamiltonian term.

\subsection{Only stochastic generators contribute linearly}
We first compute the first-order splitting of the ideal eigenvalue-one subspace $\mathcal V_1=\Span\{\Delta,\mathrm{id}\}$ defined in Eq.~\eqref{eq:V1}. Write
\begin{equation}
    \widetilde{\mathcal G}_k
    =\left(\id+\varepsilon g_k+O(\varepsilon^2)\right)\mathcal G_k,
\end{equation}
where each $g_k$ is a linear combination of the elementary generators. Under column vectorization, $\Lmat$ in Eq.~\eqref{eq:Lmatrix} is the matrix representation of the map on superoperators
\begin{equation}
    \widehat{\Lmat}(X)
    =\sum_k\Omega_k\,
      \widetilde{\mathcal G}_k\circ X\circ\mathcal G_k^{-1}.
    \label{eq:Lhat}
\end{equation}
Every $X\in\mathcal V_1$ is a linear combination of the identity and completely depolarizing superoperators, both of which commute with every ideal unitary channel. Hence, with $\bar g=\sum_k\Omega_k g_k$, the first-order perturbation acts as
\begin{equation}
    \widehat{\Lmat}_1(X)=\bar g\circ X,
    \qquad X\in\mathcal V_1.
    \label{eq:L1-slow-action}
\end{equation}
Each ideal conjugation
$X\mapsto\mathcal G_k\circ X\circ\mathcal G_k^{-1}$ is orthogonal in the Hilbert--Schmidt inner product and fixes $\mathcal V_1$ pointwise. Its adjoint therefore also fixes $\mathcal V_1$, and $\mathcal V_1^\perp$ is invariant under both the ideal moment operator and its adjoint. Since the spectrum on this complement lies strictly inside the unit disk, the spectral projector at eigenvalue one is precisely the Hilbert--Schmidt orthogonal projector $P_{\mathcal V_1}$, equal to the Haar second-moment twirl. Thus degenerate perturbation theory reduces the first-order splitting to $X\mapsto P_{\mathcal V_1}(\bar g\circ X)$.

We compute that projector explicitly with the Choi--Jamio\l{}kowski
isomorphism \cite{Choi1975,Jamiolkowski1972}.
 For any superoperator $\Phi$, let
\begin{equation}
    J(\Phi)=\sum_{i,j=0}^{d-1}\ketbra{i}{j}\otimes\Phi\!\left(\ketbra{i}{j}\right),
\end{equation}
and equip superoperators with the Hilbert--Schmidt inner product
\begin{equation}
    \langle \Phi,\Psi\rangle_{\mathrm{HS}}
    = \tr\!\left[J(\Phi)^\dagger J(\Psi)\right].
\end{equation}
The Choi matrices of $\Delta$ and $\mathrm{id}$ are
\begin{equation}
\begin{aligned}
    J(\Delta)&=\frac{\id\otimes\id}{d},\\
    J(\mathrm{id})&=d\,\ket{\Phi^+}\!\bra{\Phi^+},\\
    \ket{\Phi^+}&=\frac{1}{\sqrt d}\sum_{k=0}^{d-1}\ket{kk}.
\end{aligned}
\end{equation}

\begin{lemma}[Choi pairings for $\Delta$ and $\mathrm{id}$]\label{lem:choi-pairings}
For every Hermiticity-preserving superoperator $\Phi$,
\begin{align}
    \langle \Phi,\Delta\rangle_{\mathrm{HS}}
    &= \frac{1}{d}\tr\!\left[\Phi(\id)\right],\\
    \langle \Phi,\mathrm{id}\rangle_{\mathrm{HS}}
    &= \Tr(\Phi),
\end{align}
where $\Tr(\Phi)$ denotes the trace of $\Phi$ as a linear map on operator space.
\end{lemma}

\begin{proof}
Hermiticity preservation makes $J(\Phi)$ Hermitian. Using $J(\Delta)=\id\otimes\id/d$ gives
\begin{equation}
    \langle \Phi,\Delta\rangle_{\mathrm{HS}}
    = \frac{1}{d}\tr[J(\Phi)]
    = \frac{1}{d}\tr\!\left[\Phi(\id)\right].
\end{equation}
For the second identity choose a Hermitian orthonormal operator basis $\{E_\mu\}_{\mu=0}^{d^2-1}$. Then
\begin{equation}
    \Tr(\Phi)=\sum_\mu \tr\!\left[E_\mu\Phi(E_\mu)\right],
\end{equation}
and evaluating the Choi inner product in the same basis gives the identical sum. Hence $\langle \Phi,\mathrm{id}\rangle_{\mathrm{HS}}=\Tr(\Phi)$.
\end{proof}

\begin{proposition}[Hilbert--Schmidt projector onto $\mathcal V_1$]\label{prop:choi-projector}
Let $P_{\mathcal V_1}$ denote the Hilbert--Schmidt orthogonal projector onto $\mathcal V_1=\Span\{\Delta,\mathrm{id}\}$. For every Hermiticity-preserving superoperator $\Phi$,
\begin{equation}
P_{\mathcal V_1}(\Phi)
=
\frac{d\,\tr[\Phi(\id)]-\Tr(\Phi)}{d^2-1}\,\Delta
+
\frac{\Tr(\Phi)-\tr[\Phi(\id)]/d}{d^2-1}\,\mathrm{id}.
\label{eq:choi-projector}
\end{equation}
In particular, if $\tr[\Phi(\id)]=0$, then
\begin{equation}
P_{\mathcal V_1}(\Phi)
=
\frac{\Tr(\Phi)}{d^2-1}\,(\mathrm{id}-\Delta).
\label{eq:projector-trace-only}
\end{equation}
\end{proposition}

\begin{proof}
The Gram matrix of $\{\Delta,\mathrm{id}\}$ in the Hilbert--Schmidt inner product is
\begin{equation}
G=
\begin{pmatrix}
1 & 1\\
1 & d^2
\end{pmatrix},
\qquad
G^{-1}=
\frac{1}{d^2-1}
\begin{pmatrix}
d^2 & -1\\
-1 & 1
\end{pmatrix}.
\end{equation}
Applying $G^{-1}$ to the pairing vector from Lemma~\ref{lem:choi-pairings} gives Eq.~\eqref{eq:choi-projector}. Equation~\eqref{eq:projector-trace-only} is the special case $\tr[\Phi(\id)]=0$.
\end{proof}

We now apply Proposition~\ref{prop:choi-projector} to the four elementary generator sectors. In Liouville form they are
\begin{align}
\mathcal H_{P}
  &= -i\bigl(\id\!\otimes\!P - P^{\mathsf T}\!\otimes\!\id\bigr),
\label{eq:liouville-H}\\
\mathcal S_{P}
  &= P^{\mathsf T}\!\otimes\!P - \id\!\otimes\!\id,
\label{eq:liouville-S}\\
\mathcal C_{P,Q}
  &= Q^{\mathsf T}\!\otimes\!P + P^{\mathsf T}\!\otimes\!Q
     - \frac{1}{2}\Bigl(
         \id\!\otimes\!\{P,Q\}
         + \{P,Q\}^{\mathsf T}\!\otimes\!\id
       \Bigr),
\label{eq:liouville-C}\\
\mathcal A_{P,Q}
  &= i\bigl(Q^{\mathsf T}\!\otimes\!P-P^{\mathsf T}\!\otimes\!Q\bigr)\notag\\
  &\quad +\frac{i}{2}\bigl(
           \id\!\otimes\![P,Q]
           +[P,Q]^{\mathsf T}\!\otimes\!\id
         \bigr).
\label{eq:liouville-A}
\end{align}

\begin{proposition}[Only stochastic generators contribute linearly]\label{prop:sector-traces}
For every nonidentity Pauli $P$ and every ordered pair $P<Q$,
\begin{align}
    \tr[H_P(\id)] &= 0,
    & \Tr(H_P) &= 0,\\
    \tr[C_{P,Q}(\id)] &= 0,
    & \Tr(C_{P,Q}) &= 0,\\
    \tr[A_{P,Q}(\id)] &= 0,
    & \Tr(A_{P,Q}) &= 0,\\
    \tr[S_P(\id)] &= 0,
    & \Tr(S_P) &= -d^2.
\end{align}
Hence Hamiltonian, correlation, and active generators do not contribute to the RB decay constant at first order, while stochastic generators do.
\end{proposition}

\begin{proof}
The channel actions on the identity are immediate from the defining formulas:
\begin{align}
    H_P(\id) &= -i[P,\id]=0,\\
    S_P(\id) &= P\id P - \id = 0,\\
    C_{P,Q}(\id) &= P Q + Q P - \{P,Q\}=0,\\
    A_{P,Q}(\id) &= i\!\left(PQ-QP+[P,Q]\right)=2i[P,Q].
\end{align}
Since every commutator of Pauli operators is traceless, all four sectors satisfy $\tr[g(\id)]=0$.

It therefore remains to compute the superoperator traces. Using $\Tr(A\otimes B)=\tr(A)\tr(B)$ in Eqs.~\eqref{eq:liouville-H}--\eqref{eq:liouville-A}, together with the facts that every nonidentity Pauli is traceless and every anticommutator or commutator appearing above is either zero or proportional to a nonidentity Pauli, we obtain
\begin{align}
    \Tr(H_P) &= 0,\\
    \Tr(C_{P,Q}) &= 0,\\
    \Tr(A_{P,Q}) &= 0.
\end{align}
For the stochastic generator,
\begin{equation}
    \Tr(S_P)
    =
    \Tr(P^{\mathsf T}\!\otimes\!P)
    -
    \Tr(\id\!\otimes\!\id)
    =
    \tr(P)^2 - \tr(\id)^2
    =
    -d^2.
\end{equation}
This proves the claim.
\end{proof}

Let $D=\mathrm{id}-\Delta$. Because $\bar g$ is trace annihilating,
$\tr[\bar g(\id)]=0$, and
$\Tr(\bar g\circ\Delta)=\tr[\bar g(\id)]/d=0$.
Proposition~\ref{prop:choi-projector} gives
\begin{equation}
    P_{\mathcal V_1}(\bar g\circ\Delta)=0,\qquad
    P_{\mathcal V_1}(\bar g\circ D)
    =\frac{\Tr(\bar g)}{d^2-1}D.
\end{equation}
Thus the trace-preserving branch has no first-order shift. Writing the
depolarizing branch as
$\gamma=1+\varepsilon\delta\gamma+O(\varepsilon^2)$, its first-order
coefficient is
\begin{equation}
    \delta\gamma
    =\frac{\Tr(\bar g)}{d^2-1}
    =\frac{1}{d^2-1}\sum_k \Omega_k\,\Tr(g_k).
    \label{eq:delta-gamma-choi}
\end{equation}
Restoring the physical elementary rates, which include the bookkeeping factor $\varepsilon$, and substituting Proposition~\ref{prop:sector-traces} into Eq.~\eqref{eq:delta-gamma-choi} yields the complete linear term:
\begin{equation}
    \gamma
    =
    1-\frac{d^2}{d^2-1}
    \sum_{k,P}\Omega_k\,s_P^{(k)}
    +O(\delta^2,\delta\theta,\theta^2).
    \label{eq:gamma-linear}
\end{equation}
Equivalently, the RB error rate itself is
\begin{equation}
    r
    =
    \sum_{k,P}\Omega_k\,s_P^{(k)}
    +O(\delta^2,\delta\theta,\theta^2).
    \label{eq:r-linear}
\end{equation}
This establishes the linear part of
Theorem~\ref{thm:main-perturbative}. The entire linear dependence of RB on elementary error generators is exhausted by the stochastic sector. Hamiltonian, correlation, and active generators are invisible at first order. This does not mean that a physical noise process with an active component has no first-order RB effect. For example, single-qubit
amplitude damping has both active and stochastic components in its error generator, and its linear contribution to $r$ is supplied by the stochastic components. Complete positivity constrains the allowed combinations of error-generator coefficients jointly.

\subsection{Quadratic Hamiltonian term and gauge invariance}

To isolate the displayed order in the coherent sector, set $\delta=0$ and restrict attention to purely Hamiltonian perturbations. Equation~\eqref{eq:gamma-linear} then shows that the linear term vanishes, so the leading possible contribution is quadratic in the Hamiltonian rates:
\begin{equation}
    Q_H(\bm h;\Omega)=\bm h^{\mathsf T}\QHmat(\Omega)\bm h,
\end{equation}
for some symmetric matrix $\QHmat(\Omega)$.

Ref.~\cite{GonzalezDeMendoza2026} identifies $\Gamma_H$ as the Hamiltonian tangent space to a gate-set gauge orbit and the Hamiltonian FOGI properties as its annihilator. Under a gauge transformation, $\Lmat$ transforms by the similarity $(I\otimes\mathcal M)\Lmat(I\otimes\mathcal M^{-1})$, so $\gamma$ is exactly gauge invariant. Restricting to Hamiltonian gauge generators $K$ keeps the gates unitary.

To compare homogeneous terms without treating the gauge action as exactly linear, set $\bm h=\varepsilon\bm a$ and apply $\mathcal M=e^{\varepsilon tK}$, with $\bm a$, $K$, and $t$ fixed. Equation~\eqref{eq:lin-gauge} gives
\begin{equation}
    \bm h'
    =\varepsilon\bigl(\bm a+t\bm g_H(K)\bigr)
    +O(\varepsilon^2).
\end{equation}
There is no linear Hamiltonian term in $\gamma$, so the $O(\varepsilon^2)$ correction to $\bm h'$ can affect $\gamma$ only at order $\varepsilon^3$ or higher. Comparing the coefficients of $\varepsilon^2$ in the exactly equal quantities $\gamma(\bm h')$ and $\gamma(\bm h)$ therefore gives
\begin{equation}
    Q_H(\bm a+t\bm g_H(K);\Omega)=Q_H(\bm a;\Omega).
\end{equation}
The coefficient linear in $t$ is
$2\bm a^{\mathsf T}\QHmat(\Omega)\bm g_H(K)$.
It vanishes for every $\bm a$, so
\begin{equation}
    \QHmat(\Omega)\,\bm g_H(K)=0.
    \label{eq:gauge-kernel}
\end{equation}
Thus $Q_H$ is a well-defined quadratic form on $\mathcal H/\Gamma_H$. This is the precise sense in which RB depends only on gauge-invariant combinations of coherent errors. Combining this with Eq.~\eqref{eq:rfromgamma}, we may also write the purely Hamiltonian RB error rate as
\begin{equation}
    r
    =
    \frac{d^2-1}{d^2}\,Q_H(\bm h;\Omega)
    +
    O(\theta^3).
\label{eq:r-hamiltonian}
\end{equation}

\section{Direct RB for the \texorpdfstring{$\{X_{\pi/2},Y_{\pi/2}\}$}{\{Xpi/2,Ypi/2\}} gate set}
\label{sec:drb-xy}
For a general gate set we do not know a useful closed form for $\QHmat(\Omega)$. We therefore now specialize to the two-gate example of Sec.~\ref{sec:motivation}, where the \FOGI reduction can be carried out analytically. Let
\begin{equation}
    p=\Omega(X_{\pi/2}),\qquad q=1-p=\Omega(Y_{\pi/2}).
\end{equation}
Appendix~\ref{app:xy-mixing} proves that this gate set is a sequence-asymptotic unitary 2-design for every fixed $0<p<1$, but not at the endpoints, and that its moment gap closes linearly there. For the Hamiltonian sector of the $\{X_{\pi/2},Y_{\pi/2}\}$ gate set, the full $6\times 6$ quadratic-form matrix of $1-\gamma$ in the coordinates
\begin{equation}
\vec{h} = \bigl(h_X^{(X)},h_Y^{(X)},h_Z^{(X)},h_X^{(Y)},h_Y^{(Y)},h_Z^{(Y)}\bigr)
\end{equation}
has rank $3$ for every $0<p<1$. Its kernel is exactly the three-dimensional linearized gauge subspace. Passing to the FOGI coordinates $\Fvec=(F_1,F_2,F_3)^T$ of Eqs.~\eqref{eq:F1}--\eqref{eq:F3}, which are the Hamiltonian specialization of Ref.~\cite{GonzalezDeMendoza2026}, therefore removes the gauge directions, and the reduced quadratic form is
\begin{equation}
    1-\gamma
    = \Fvec^{\,T}\Qmat_{\mathrm{DRB}}(p)\Fvec + O\!\left(\norm{\bm h}^3\right).
    \label{eq:drb-quad}
\end{equation}

\begin{proposition}\label{prop:qdrb}
For $0<p<1$, the reduced Hamiltonian quadratic-form matrix of $1-\gamma$ on the \FOGI quotient of the $\{X_{\pi/2},Y_{\pi/2}\}$ gate set is
\begin{equation}
\Qmat_{\mathrm{DRB}}(p)=
\frac{1}{3}
\begin{pmatrix}
\displaystyle \frac{4p(1+p^2)}{q} & 4pq & 8p^2\\[8pt]
4pq & \displaystyle \frac{4q(1+q^2)}{p} & 8q^2\\[8pt]
8p^2 & 8q^2 & 16pq
\end{pmatrix}.
\label{eq:Qdrb}
\end{equation}
This $3\times 3$ matrix is positive definite for every $0<p<1$.
\end{proposition}

\begin{proof}
Appendix~\ref{app:drb-manual} derives the full $6\times 6$ quadratic-form matrix of $1-\gamma$ from second-order perturbation theory on the exact one-qubit $\Lmat$ matrix. That matrix annihilates the three gauge vectors generated by Eq.~\eqref{eq:lin-gauge}. Choosing any representatives of the three \FOGI coordinates $(F_1,F_2,F_3)$ and evaluating the quadratic form on those representatives yields Eq.~\eqref{eq:Qdrb}; the result is independent of the chosen representatives because the full quadratic-form matrix vanishes on the gauge subspace. Positive definiteness for $0<p<1$ follows from Sylvester's criterion, since the leading principal minors satisfy
\begin{equation}
\begin{aligned}
Q_{11}=\frac{4p(1+p^2)}{3q}>0,
\\
\det \Qmat_{\mathrm{DRB}}^{(2)}=\frac{32(p^2-p+1)}{9}>0,
\\
\det\Qmat_{\mathrm{DRB}}=\frac{256pq}{27}>0.
\end{aligned}
\end{equation}
\end{proof}

This is the central analytic result for the example gate set. It captures the coherent dependence through quadratic order, since the remaining three Hamiltonian directions are precisely the linearized gauge directions. All big-$O$ statements in this section are at fixed $p\in(0,1)$ and refer to the full Hamiltonian-rate vector $\bm h$, not just $\Fvec$. For the one-qubit RB error rate itself, Eq.~\eqref{eq:rfromgamma} gives
\begin{equation}
    r_{\mathrm{DRB}}(p)
    =
    \frac{3}{4}\,\Fvec^{\,T}\Qmat_{\mathrm{DRB}}(p)\Fvec
    + O\!\left(\norm{\bm h}^3\right).
    \label{eq:drb-r}
\end{equation}

Several basic features are immediate. First, the coefficients diverge as $p\to0$ or $1$. This is a singularity of the fixed-$p$ perturbative coefficient, not a divergence of the exact physical RB rate, which remains bounded; at fixed error strength the expansion fails before the endpoint. At $p=0$ or $1$ the sampled gate set is not a sequence-asymptotic 2-design. Second, the relational coordinate $F_3$ is weighted by $16pq/3$, so its sensitivity is maximal near $p=\nicefrac{1}{2}$ and vanishes as either gate disappears. Third, the intrinsic coordinates $F_1$ and $F_2$ are weighted asymmetrically when $p\neq\nicefrac{1}{2}$.

At the symmetric point $p=\nicefrac{1}{2}$,
\begin{equation}
\Qmat_{\mathrm{DRB}}\!\left(\frac{1}{2}\right)
=
\frac{1}{3}
\begin{pmatrix}
5 & 1 & 2\\
1 & 5 & 2\\
2 & 2 & 4
\end{pmatrix},
\end{equation}
whose normalized eigenvectors in the displayed $(F_1,F_2,F_3)$ coordinates are
\begin{align}
v_1 &= (1,1,1)^T/\sqrt3, & \lambda_1 &= \frac{8}{3},\\
v_2 &= (1,-1,0)^T/\sqrt2, & \lambda_2 &= \frac{4}{3},\\
v_3 &= (-1,-1,2)^T/\sqrt6, & \lambda_3 &= \frac{2}{3}.
\end{align}
With the Euclidean metric in these coordinates, the most sensitive mode has equal $F_1,F_2,F_3$, the next is differential over-rotation, and the least sensitive mode has $F_1=F_2=-F_3/2$, in which relational tilt partly cancels equal intrinsic over-rotations.

The sampling dependence of these eigenvalues is illustrated in Fig.~\ref{fig:drb_eigs} in Appendix~\ref{app:drb-spectrum}.

For the coherent error models introduced in Sec.~\ref{sec:motivation}, where $F_1=\theta_X/2+O(\eta^2)$, $F_2=\theta_Y/2+O(\eta^2)$, and $F_3=\phi/2+O(\eta^2)$, Eq.~\eqref{eq:drb-r} yields explicit formulas for the leading RB error. For equal over-rotation or equal under-rotation,
\begin{equation}
r_{\mathrm{eq}}(p)
=
\frac{(p^2-p+1)(2p^2-2p+1)}{2p(1-p)}\,\theta^2
+
O(\theta^3),
\label{eq:req}
\end{equation}
for opposite over- and under-rotation,
\begin{equation}
r_{\mathrm{opp}}(p)
=
\frac{3p^2-3p+1}{2p(1-p)}\,\theta^2
+
O(\theta^3),
\label{eq:ropp}
\end{equation}
and for pure tilt,
\begin{equation}
r_{\mathrm{tilt}}(p)
=
p(1-p)\phi^2
+
O(\phi^3).
\label{eq:rtilt}
\end{equation}
For a quantitative comparison with native-gate infidelity, consider equal over-rotations, $\theta_X=\theta_Y=\theta$ and $\phi=0$. Both native gates have the same leading entanglement infidelity $\theta^2/4$, so the weighted-infidelity heuristic also gives $\theta^2/4$, independently of $p$. At the symmetric sampling probability, Eq.~\eqref{eq:req} gives
\begin{equation}
    r_{\mathrm{eq}}\!\left(\frac12\right)
    =\frac34\theta^2+O(\theta^3),
\end{equation}
which is three times the native-gate infidelity to leading order. At $p=0.9$, the leading-order ratio is approximately $16.6$. These comparisons take the small-error limit at each fixed sampling probability; they do not assume a perturbative approximation that remains uniform as $p$ approaches an endpoint.

\section{\CRB as a compilation-dependent quadratic form}
\label{sec:crb}

We previously saw in Figure~\ref{fig:crb_example} that, even when native-gate infidelities are equal to leading order, the mean-length-normalized RB rate can vary substantially with both the kind of coherent error and the chosen compilation. To quantitatively explain this effect, we apply the same perturbative framework to the uniformly weighted one-qubit Clifford group. For a fixed compilation $\kappa$, let a Clifford $C\in C_1$ be implemented by a native sequence $g_1,\ldots,g_m$ with $g_j\in\{X_{\pi/2},Y_{\pi/2}\}$. Refactoring the compiled Clifford into post-gate form gives, to first order in the native Hamiltonian errors,
\begin{equation}
    \mathbf h_C^{(\kappa)} = \sum_{j=1}^m R_{g_m\cdots g_{j+1}}\,\mathbf h_{g_j},
    \label{eq:compiled-hC}
\end{equation}
where $\mathbf h_{g_j}$ is the Hamiltonian vector of the $j$th native gate and $R_{g_m\cdots g_{j+1}}$ is the Bloch-sphere rotation induced by the full ideal suffix following it (the identity when $j=m$). Here $\mathbf h_C^{(\kappa)}$ denotes the first-order compiled vector. Its difference from the exact compiled Hamiltonian vector is second order in the native error coefficients and therefore affects $\gamma$ only at third order or higher, since the linear Hamiltonian contribution vanishes. Thus each compilation determines a point $\mathbf h=(\mathbf h_C)_{C\in C_1}\in V:=(\mathbb R^3)^{24}$. For the spectral statements below we equip this coordinate space with
\begin{equation}
    \langle\mathbf h,\mathbf k\rangle_V=\sum_{C\in C_1}\mathbf h_C\cdot\mathbf k_C.
    \label{eq:crb-metric}
\end{equation}
This is the Euclidean metric of the Hamiltonian coefficients. The gauge kernel and quotient are canonical, but orthogonality and quadratic-form eigenvalues depend on this metric \cite{GonzalezDeMendoza2026}.

On the traceless Pauli sector, the CRB $\Lmat$ matrix is
\begin{equation}
\mathscr L(\mathbf h)=\frac1{24}\sum_{C\in C_1} R_C\otimes e^{B_C}R_C,
\qquad
B_C:=2[\mathbf h_C]_\times .
\label{eq:crb-transfer}
\end{equation}
Writing $v=\operatorname{vec}(I_3)/\sqrt 3$, $P=\ket{v}\!\bra{v}$, and
\begin{equation}
\mathbf u(\mathbf h):=\sum_{C\in C_1}\mathbf h_C,
\qquad
\mathbf w(\mathbf h):=\sum_{C\in C_1}R_C^{\mathsf T}\mathbf h_C,
\label{eq:u-w-main}
\end{equation}
we obtain the following quadratic description of coherent CRB.

\begin{proposition}\label{prop:crb-main}
Let $\gamma(\mathbf h)$ be the CRB decay eigenvalue associated with Eq.~\eqref{eq:crb-transfer}, and define its homogeneous quadratic coefficient by
\begin{equation}
1-\gamma(\mathbf h)=Q_{\mathrm{CRB}}^{(2)}(\mathbf h)+O(\|\mathbf h\|^3).
\end{equation}
Then
\begin{equation}
Q_{\mathrm{CRB}}^{(2)}(\mathbf h)
=
\frac1{18}\sum_{C\in C_1}\|\mathbf h_C\|^2
+
\frac1{216}\,\mathbf u(\mathbf h)\!\cdot\!\mathbf w(\mathbf h).
\label{eq:crb-quadratic-form}
\end{equation}
Equivalently, if
\[
V_{\rm gauge}=\{\mathbf h_C=\mathbf a-R_C\mathbf a:\mathbf a\in\mathbb R^3\},
\]
\[
V_{\rm dg}=\{\mathbf h_C=\mathbf a+R_C\mathbf a:\mathbf a\in\mathbb R^3\},
\]
\[
V_\perp=\{\mathbf h\in V:\mathbf u(\mathbf h)=0,\ \mathbf w(\mathbf h)=0\},
\]
then
\begin{equation}
D_{\rm CRB}=\frac19\,\Pi_{\rm dg}+\frac1{18}\,\Pi_\perp,
\qquad
\ker D_{\rm CRB}=V_{\rm gauge},
\label{eq:Dcrb-main}
\end{equation}
where $D_{\rm CRB}$ is the quadratic-form matrix defined by $Q_{\mathrm{CRB}}^{(2)}(\mathbf h)=\mathbf h^{\mathsf T}D_{\rm CRB}\mathbf h$.
\end{proposition}

\begin{proof}
A direct perturbative derivation, together with a complementary symmetry interpretation, is given in Appendix~\ref{app:crb-manual}.
\end{proof}

For one qubit, the associated RB error rate is simply
\begin{equation}
    r_{\mathrm{CRB}}=\frac34\,Q_{\mathrm{CRB}}^{(2)}(\mathbf h)+O(\|\mathbf h\|^3).
\end{equation}
Thus the dashed estimates in Fig.~\ref{fig:crb_example} plot $\frac{3}{4\bar m_\kappa}Q_{\mathrm{CRB}}^{(2)}(\mathbf h)$ for each compilation, while the filled histograms use the exact compiled-Clifford $\Lmat$ matrix.

Equation~\eqref{eq:Dcrb-main} shows that, to quadratic order, \CRB is insensitive to the gauge subspace, weights the three symmetric collective directions by $1/9$, and weights the remaining 66 orthogonal directions by $1/18$. Combined with Eq.~\eqref{eq:compiled-hC}, this explains Fig.~\ref{fig:crb_example}: compilation changes how a native error vector projects into these fixed Clifford-level sectors.

\section{Conclusion}

We have developed a perturbative, gauge-invariant interpretation of RB for small gate-local Markovian noise on random circuits. We have shown that for every fixed sequence-asymptotic unitary 2-design and in a small-noise regime, stochastic Pauli error rates contribute linearly to the RB decay parameter, while Hamiltonian error rates contribute through a gauge-invariant quadratic form. Gauge invariance forces that quadratic form to annihilate the linearized gauge subspace, so the leading coherent RB contribution depends only on \FOGI properties \cite{GonzalezDeMendoza2026}. For the native gate set $\{X_{\pi/2},Y_{\pi/2}\}$, the Hamiltonian sector reduces to three \FOGI coordinates: two intrinsic over-rotation coordinates and one relational tilt coordinate. The resulting analytic formula explains both the sampling dependence of \DRB and the compilation dependence of \CRB.

The explicit \DRB formulas derived here are useful in their own right. Measurements at several sampling probabilities give quadratic constraints on intrinsic over-rotations and relational axis tilt; resolving signed error coordinates requires additional calibration information or controlled error injections. The formulas therefore predict the RB response to a specified error model, rather than reconstructing a full model from one scalar rate.

Several natural extensions remain. First, it would be valuable to derive similarly explicit quadratic forms for larger native gate sets and for multi-qubit gate sets, containing multi-qubit entangling gates. Second, correlation and active generators should be incorporated beyond the first nontrivial order treated here. Finally, the same perturbative viewpoint may be useful for other random-circuit diagnostics (e.g., Pauli observables), but their response coefficients must be derived for each observable.

\section*{Acknowledgments}
Work partially funded by NSF award DMR-1747426. AMP acknowledges funding from a NASA Space Technology Graduate Research Opportunity award. This material was funded in part by the U.S. Department of Energy, Office of Science, Office of Advanced Scientific Computing Research Quantum Testbed Program, as well as by Sandia National Laboratories' Laboratory Directed Research and Development program. Sandia National Laboratories is a multimission laboratory managed and operated by National Technology and Engineering Solutions of Sandia, LLC, a wholly owned subsidiary of Honeywell International, Inc., for the U.S. Department of Energy's National Nuclear Security Administration under contract DE-NA0003525.
\bibliography{bib}

@article{Knill2008,
  author  = {Knill, E. and Leibfried, D. and Reichle, R. and Britton, J. and Blakestad, R. B. and Jost, J. D. and Langer, C. and Ozeri, R. and Seidelin, S. and Wineland, D. J.},
  title   = {Randomized Benchmarking of Quantum Gates},
  journal = {Phys. Rev. A},
  volume  = {77},
  pages   = {012307},
  year    = {2008},
  doi     = {10.1103/PhysRevA.77.012307}
}

@article{Magesan2011,
  author  = {Magesan, Easwar and Gambetta, J. M. and Emerson, Joseph},
  title   = {Scalable and Robust Randomized Benchmarking of Quantum Processes},
  journal = {Phys. Rev. Lett.},
  volume  = {106},
  pages   = {180504},
  year    = {2011},
  doi     = {10.1103/PhysRevLett.106.180504}
}

@article{Magesan2012,
  author  = {Magesan, Easwar and Gambetta, Jay M. and Emerson, Joseph},
  title   = {Characterizing Quantum Gates via Randomized Benchmarking},
  journal = {Phys. Rev. A},
  volume  = {85},
  pages   = {042311},
  year    = {2012},
  doi     = {10.1103/PhysRevA.85.042311}
}

@article{Proctor2017,
  author  = {Proctor, Timothy and Rudinger, Kenneth and Young, Kevin and Sarovar, Mohan and Blume-Kohout, Robin},
  title   = {What Randomized Benchmarking Actually Measures},
  journal = {Phys. Rev. Lett.},
  volume  = {119},
  pages   = {130502},
  year    = {2017},
  doi     = {10.1103/PhysRevLett.119.130502}
}

@article{Wallman2018,
  author  = {Wallman, Joel J.},
  title   = {Randomized Benchmarking with Gate-Dependent Noise},
  journal = {Quantum},
  volume  = {2},
  pages   = {47},
  year    = {2018},
  doi     = {10.22331/q-2018-01-29-47}
}

@article{CarignanDugas2018,
  author  = {Carignan-Dugas, Arnaud and Boone, Kristine and Wallman, Joel J. and Emerson, Joseph},
  title   = {From Randomized Benchmarking Experiments to Gate-Set Circuit Fidelity: How to Interpret Randomized Benchmarking Decay Parameters},
  journal = {New J. Phys.},
  volume  = {20},
  pages   = {092001},
  year    = {2018},
  doi     = {10.1088/1367-2630/aadcc7}
}

@article{Boone2019,
  author  = {Boone, Kristine and Carignan-Dugas, Arnaud and Wallman, Joel J. and Emerson, Joseph},
  title   = {Randomized Benchmarking under Different Gate Sets},
  journal = {Phys. Rev. A},
  volume  = {99},
  pages   = {032329},
  year    = {2019},
  doi     = {10.1103/PhysRevA.99.032329}
}

@article{Proctor2019,
  author  = {Proctor, Timothy J. and Carignan-Dugas, Arnaud and Rudinger, Kenneth and Nielsen, Erik and Blume-Kohout, Robin and Young, Kevin},
  title   = {Direct Randomized Benchmarking for Multiqubit Devices},
  journal = {Phys. Rev. Lett.},
  volume  = {123},
  pages   = {030503},
  year    = {2019},
  doi     = {10.1103/PhysRevLett.123.030503}
}

@article{Nielsen2021,
  author  = {Nielsen, Erik and Gamble, John King and Rudinger, Kenneth and Scholten, Travis and Young, Kevin and Blume-Kohout, Robin},
  title   = {Gate Set Tomography},
  journal = {Quantum},
  volume  = {5},
  pages   = {557},
  year    = {2021},
  doi     = {10.22331/q-2021-10-05-557}
}

@article{BlumeKohout2022,
  author  = {Blume-Kohout, Robin and da Silva, Marcus P. and Nielsen, Erik and Proctor, Timothy and Rudinger, Kenneth and Sarovar, Mohan and Young, Kevin},
  title   = {A Taxonomy of Small Markovian Errors},
  journal = {PRX Quantum},
  volume  = {3},
  pages   = {020335},
  year    = {2022},
  doi     = {10.1103/PRXQuantum.3.020335}
}

@article{Madzik2022,
  author  = {M{\k{a}}dzik, Mateusz T. and Asaad, Sameh and Youssry, Ahmed and Joecker, Boris and Rudinger, Kenneth M. and Nielsen, Erik and Young, Kevin C. and Proctor, Timothy J. and Baczewski, Andrew D. and Laucht, Andreas and Schmitt, Vincent and Hudson, Fleur E. and Itoh, Kohei M. and Jakob, Andrea M. and Johnson, Brett C. and Jamieson, David N. and Dzurak, Andrea S. and Ferrie, Christopher and Blume-Kohout, Robin and Morello, Andrea},
  title   = {Precision Tomography of a Three-Qubit Donor Quantum Processor in Silicon},
  journal = {Nature},
  volume  = {601},
  pages   = {348--353},
  year    = {2022},
  doi     = {10.1038/s41586-021-04292-7}
}

@article{Polloreno2025,
  author  = {Polloreno, Anthony M. and Carignan-Dugas, Arnaud and Hines, Jordan and Blume-Kohout, Robin and Young, Kevin and Proctor, Timothy},
  title   = {A Theory of Direct Randomized Benchmarking},
  journal = {Quantum},
  volume  = {9},
  pages   = {1848},
  year    = {2025},
  doi     = {10.22331/q-2025-09-05-1848}
}

@article{polloreno2022robustly,
  title={Robustly decorrelating errors with mixed quantum gates},
  author={Polloreno, Anthony M and Young, Kevin C},
  journal={Quantum Science \& Technology},
  volume={7},
  number={2},
  pages={025004},
  year={2022},
  publisher={IOP Publishing}
}

@article{Hashim2025,
  author  = {Hashim, Akel and Nguyen, Long B. and Goss, Noah and Marinelli, Brian and Naik, Ravi K. and Chistolini, Trevor and Hines, Jordan and Marceaux, J. P. and Kim, Yosep and Gokhale, Pranav and Tomesh, Teague and Chen, Senrui and Jiang, Liang and Ferracin, Samuele and Rudinger, Kenneth and Proctor, Timothy and Young, Kevin C. and Siddiqi, Irfan and Blume-Kohout, Robin},
  title   = {A Practical Introduction to Benchmarking and Characterization of Quantum Computers},
  journal = {PRX Quantum},
  volume  = {6},
  pages   = {030202},
  year    = {2025},
  doi     = {10.1103/PRXQuantum.6.030202}
}

@article{Merkel2013,
  author = {Merkel, Seth T. and Gambetta, Jay M. and Smolin, John A. and Poletto, Stefano and Corcoles, Antonio D. and Johnson, Blake R. and Ryan, Colm A. and Steffen, Matthias},
  title = {Self-consistent quantum process tomography},
  journal = {Phys. Rev. A},
  volume = {87},
  pages = {062119},
  year = {2013},
  doi = {10.1103/PhysRevA.87.062119}
}

@article{BlumeKohout2017,
  author = {Blume-Kohout, Robin and Gamble, John King and Nielsen, Erik and Rudinger, Kenneth and Mizrahi, Jonathan and Fortier, Kevin and Maunz, Peter},
  title = {Demonstration of qubit operations below a rigorous fault tolerance threshold with gate set tomography},
  journal = {Nat. Commun.},
  volume = {8},
  pages = {14485},
  year = {2017},
  doi = {10.1038/ncomms14485}
}

@article{Kelly2014,
  author = {Kelly, J. and Barends, R. and Campbell, B. and Chen, Y. and Chen, Z. and Chiaro, B. and Dunsworth, A. and Fowler, A. G. and Hoi, I.-C. and Jeffrey, E. and others},
  title = {Optimal quantum control using randomized benchmarking},
  journal = {Phys. Rev. Lett.},
  volume = {112},
  pages = {240504},
  year = {2014},
  doi = {10.1103/PhysRevLett.112.240504}
}

@article{Rol2017,
  author = {Rol, M. A. and Bultink, C. C. and O'Brien, T. E. and de Jong, S. R. and Theis, L. S. and Fu, X. and Luthi, F. and Vermeulen, R. F. L. and de Sterke, J. C. and Bruno, A. and others},
  title = {Restless tuneup of high-fidelity qubit gates},
  journal = {Phys. Rev. Applied},
  volume = {7},
  pages = {041001},
  year = {2017},
  doi = {10.1103/PhysRevApplied.7.041001}
}

@article{Erhard2019,
  author = {Erhard, Alexander and Wallman, Joel James and Postler, Lukas and Meth, Michael and Stricker, Roman and Martinez, Esteban Adrian and Schindler, Philipp and Monz, Thomas and Emerson, Joseph and Blatt, Rainer},
  title = {Characterizing large-scale quantum computers via cycle benchmarking},
  journal = {Nat. Commun.},
  volume = {10},
  pages = {5347},
  year = {2019},
  doi = {10.1038/s41467-019-13068-7}
}

@article{Proctor2021,
  author = {Proctor, Timothy and Rudinger, Kenneth and Young, Kevin and Nielsen, Erik and Blume-Kohout, Robin},
  title = {Measuring the capabilities of quantum computers},
  journal = {Nat. Phys.},
  volume = {18},
  pages = {75--79},
  year = {2022},
  doi = {10.1038/s41567-021-01409-7}
}

@article{MagesanInterleaved2012,
  author = {Magesan, Easwar and Gambetta, Jay M. and Johnson, Blake R. and Ryan, Colm A. and Chow, Jerry M. and Merkel, Seth T. and da Silva, Marcus P. and Keefe, George A. and Rothwell, Mary B. and Ohki, Thomas A. and others},
  title = {Efficient measurement of quantum gate error by interleaved randomized benchmarking},
  journal = {Phys. Rev. Lett.},
  volume = {109},
  pages = {080505},
  year = {2012},
  doi = {10.1103/PhysRevLett.109.080505}
}

@article{Gambetta2012Simultaneous,
  author = {Gambetta, Jay M. and Corcoles, Antonio D. and Merkel, Seth T. and Johnson, Blake R. and Smolin, John A. and Chow, Jerry M. and Ryan, Colm A. and Rigetti, Chad and Poletto, Stefano and Ohki, Thomas A. and others},
  title = {Characterization of addressability by simultaneous randomized benchmarking},
  journal = {Phys. Rev. Lett.},
  volume = {109},
  pages = {240504},
  year = {2012},
  doi = {10.1103/PhysRevLett.109.240504}
}

@misc{McKay2020,
  author = {McKay, David C. and Cross, Andrew W. and Wood, Christopher J. and Gambetta, Jay M.},
  title = {Correlated randomized benchmarking},
  year = {2020},
  eprint = {2003.02354},
  archivePrefix = {arXiv},
  primaryClass = {quant-ph}
}

@article{WallmanLoss2015,
  author = {Wallman, Joel J. and Barnhill, Marie and Emerson, Joseph},
  title = {Robust characterization of loss rates},
  journal = {Phys. Rev. Lett.},
  volume = {115},
  pages = {060501},
  year = {2015},
  doi = {10.1103/PhysRevLett.115.060501}
}

@article{Chasseur2015,
  author = {Chasseur, Tobias and Wilhelm, Frank K.},
  title = {Complete randomized benchmarking protocol accounting for leakage errors},
  journal = {Phys. Rev. A},
  volume = {92},
  pages = {042333},
  year = {2015},
  doi = {10.1103/PhysRevA.92.042333}
}

@article{Wood2018,
  author = {Wood, Christopher J. and Gambetta, Jay M.},
  title = {Quantification and characterization of leakage errors},
  journal = {Phys. Rev. A},
  volume = {97},
  pages = {032306},
  year = {2018},
  doi = {10.1103/PhysRevA.97.032306}
}

@article{WallmanCoherence2015,
  author = {Wallman, Joel J. and Granade, Christopher and Harper, Robin and Flammia, Steven T.},
  title = {Estimating the coherence of noise},
  journal = {New J. Phys.},
  volume = {17},
  pages = {113020},
  year = {2015},
  doi = {10.1088/1367-2630/17/11/113020}
}

@article{Feng2016,
  author = {Feng, Guanru and Wallman, Joel J. and Buonacorsi, Brandon and Cho, Franklin H. and Park, Daniel K. and Xin, Tao and Lu, Dawei and Baugh, Jonathan and Laflamme, Raymond},
  title = {Estimating the coherence of noise in quantum control of a solid-state qubit},
  journal = {Phys. Rev. Lett.},
  volume = {117},
  pages = {260501},
  year = {2016},
  doi = {10.1103/PhysRevLett.117.260501}
}

@article{Sheldon2016,
  author = {Sheldon, Sarah and Bishop, Lev S. and Magesan, Easwar and Filipp, Stefan and Chow, Jerry M. and Gambetta, Jay M.},
  title = {Characterizing errors on qubit operations via iterative randomized benchmarking},
  journal = {Phys. Rev. A},
  volume = {93},
  pages = {012301},
  year = {2016},
  doi = {10.1103/PhysRevA.93.012301}
}

@article{CarignanDugas2015,
  author = {Carignan-Dugas, Arnaud and Wallman, Joel J. and Emerson, Joseph},
  title = {Characterizing universal gate sets via dihedral benchmarking},
  journal = {Phys. Rev. A},
  volume = {92},
  pages = {060302},
  year = {2015},
  doi = {10.1103/PhysRevA.92.060302}
}

@article{Cross2016,
  author = {Cross, Andrew W. and Magesan, Easwar and Bishop, Lev S. and Smolin, John A. and Gambetta, Jay M.},
  title = {Scalable randomized benchmarking of non-Clifford gates},
  journal = {npj Quantum Inf.},
  volume = {2},
  pages = {16012},
  year = {2016},
  doi = {10.1038/npjqi.2016.12}
}

@article{BrownEastin2018,
  author = {Brown, Winton G. and Eastin, Bryan},
  title = {Randomized benchmarking with restricted gate sets},
  journal = {Phys. Rev. A},
  volume = {97},
  pages = {062323},
  year = {2018},
  doi = {10.1103/PhysRevA.97.062323}
}

@article{Hashagen2018,
  author = {Hashagen, A. K. and Flammia, Steven T. and Gross, David and Wallman, Joel J.},
  title = {Real randomized benchmarking},
  journal = {Quantum},
  volume = {2},
  pages = {85},
  year = {2018},
  doi = {10.22331/q-2018-08-22-85}
}

@article{Helsen2019,
  author = {Helsen, Jonas and Xue, Xiao and Vandersypen, Lieven M. K. and Wehner, Stephanie},
  title = {A new class of efficient randomized benchmarking protocols},
  journal = {npj Quantum Inf.},
  volume = {5},
  pages = {71},
  year = {2019},
  doi = {10.1038/s41534-019-0182-7}
}

@article{Claes2021,
  author = {Claes, Jahan and Rieffel, Eleanor and Wang, Zhihui},
  title = {Character randomized benchmarking for non-multiplicity-free groups with applications to subspace, leakage, and matchgate randomized benchmarking},
  journal = {PRX Quantum},
  volume = {2},
  pages = {010351},
  year = {2021},
  doi = {10.1103/PRXQuantum.2.010351}
}

@article{Helsen2022General,
  author = {Helsen, J. and Roth, I. and Onorati, E. and Werner, A. H. and Eisert, J.},
  title = {General framework for randomized benchmarking},
  journal = {PRX Quantum},
  volume = {3},
  pages = {020357},
  year = {2022},
  doi = {10.1103/PRXQuantum.3.020357}
}

@article{Helsen2022Matchgate,
  author = {Helsen, Jonas and Nezami, Sepehr and Reagor, Matthew and Walter, Michael},
  title = {Matchgate benchmarking: Scalable benchmarking of a continuous family of many-qubit gates},
  journal = {Quantum},
  volume = {6},
  pages = {657},
  year = {2022},
  doi = {10.22331/q-2022-02-21-657}
}

@article{FrancaHashagen2018,
  author = {Franca, D. S. and Hashagen, A. K.},
  title = {Approximate randomized benchmarking for finite groups},
  journal = {J. Phys. A: Math. Theor.},
  volume = {51},
  pages = {395302},
  year = {2018},
  doi = {10.1088/1751-8121/aad6fa}
}

@article{Hines2023,
  author = {Hines, Jordan and Lu, Marie and Naik, Ravi K. and Hashim, Akel and Ville, Jean-Loup and Mitchell, Brad and Kriekebaum, John Mark and Santiago, David I. and Seritan, Stefan and Nielsen, Erik and others},
  title = {Demonstrating scalable randomized benchmarking of universal gate sets},
  journal = {Phys. Rev. X},
  volume = {13},
  pages = {041030},
  year = {2023},
  doi = {10.1103/PhysRevX.13.041030}
}

@article{Hines2024,
  author = {Hines, Jordan and Hothem, Daniel and Blume-Kohout, Robin and Whaley, Birgitta and Proctor, Timothy},
  title = {Fully scalable randomized benchmarking without motion reversal},
  journal = {PRX Quantum},
  volume = {5},
  pages = {030334},
  year = {2024},
  doi = {10.1103/PRXQuantum.5.030334}
}

@misc{Mayer2021,
  author = {Mayer, Karl and Hall, Alex and Gatterman, Thomas and Halit, Si Khadir and Lee, Kenny and Bohnet, Justin and Gresh, Dan and Hankin, Aaron and Gilmore, Kevin and Gaebler, John},
  title = {Theory of mirror benchmarking and demonstration on a quantum computer},
  year = {2021},
  eprint = {2108.10431},
  archivePrefix = {arXiv},
  primaryClass = {quant-ph},
  doi = {10.48550/arXiv.2108.10431}
}

@misc{Heinrich2022,
  author = {Heinrich, Markus and Kliesch, Martin and Roth, Ingo},
  title = {Randomized benchmarking with random quantum circuits},
  year = {2022},
  eprint = {2212.06181},
  archivePrefix = {arXiv},
  primaryClass = {quant-ph}
}

@article{WallmanFlammia2014,
  author = {Wallman, Joel J. and Flammia, Steven T.},
  title = {Randomized benchmarking with confidence},
  journal = {New J. Phys.},
  volume = {16},
  pages = {103032},
  year = {2014},
  doi = {10.1088/1367-2630/16/10/103032}
}

@article{Epstein2014,
  author = {Epstein, Jeffrey M. and Cross, Andrew W. and Magesan, Easwar and Gambetta, Jay M.},
  title = {Investigating the limits of randomized benchmarking protocols},
  journal = {Phys. Rev. A},
  volume = {89},
  pages = {062321},
  year = {2014},
  doi = {10.1103/PhysRevA.89.062321}
}

@article{Granade2015,
  author = {Granade, Christopher and Ferrie, Christopher and Cory, D. G.},
  title = {Accelerated randomized benchmarking},
  journal = {New J. Phys.},
  volume = {17},
  pages = {013042},
  year = {2015},
  doi = {10.1088/1367-2630/17/1/013042}
}

@misc{Hincks2018,
  author = {Hincks, Ian and Wallman, Joel J. and Ferrie, Chris and Granade, Chris and Cory, David G.},
  title = {Bayesian inference for randomized benchmarking protocols},
  year = {2018},
  eprint = {1802.00401},
  archivePrefix = {arXiv},
  primaryClass = {quant-ph}
}

@article{Harper2019,
  author = {Harper, Robin and Hincks, Ian and Ferrie, Chris and Flammia, Steven T. and Wallman, Joel J.},
  title = {Statistical analysis of randomized benchmarking},
  journal = {Phys. Rev. A},
  volume = {99},
  pages = {052350},
  year = {2019},
  doi = {10.1103/PhysRevA.99.052350}
}

@article{Fogarty2015,
  author = {Fogarty, M. A. and Veldhorst, M. and Harper, R. and Yang, C. H. and Bartlett, S. D. and Flammia, S. T. and Dzurak, A. S.},
  title = {Nonexponential fidelity decay in randomized benchmarking with low-frequency noise},
  journal = {Phys. Rev. A},
  volume = {92},
  pages = {022326},
  year = {2015},
  doi = {10.1103/PhysRevA.92.022326}
}

@article{Mavadia2018,
  author = {Mavadia, S. and Edmunds, C. L. and Hempel, C. and Ball, H. and Roy, F. and Stace, T. M. and Biercuk, M. J.},
  title = {Experimental quantum verification in the presence of temporally correlated noise},
  journal = {npj Quantum Inf.},
  volume = {4},
  pages = {7},
  year = {2018},
  doi = {10.1038/s41534-017-0052-0}
}

@article{Proctor2020,
  author = {Proctor, Timothy and Revelle, Melissa and Nielsen, Erik and Rudinger, Kenneth and Lobser, Daniel and Maunz, Peter and Blume-Kohout, Robin and Young, Kevin},
  title = {Detecting and tracking drift in quantum information processors},
  journal = {Nat. Commun.},
  volume = {11},
  pages = {5396},
  year = {2020},
  doi = {10.1038/s41467-020-19074-4}
}

@article{Merkel2021,
  author = {Merkel, Seth T. and Pritchett, Emily J. and Fong, Bryan H.},
  title = {Randomized benchmarking as convolution: Fourier analysis of gate dependent errors},
  journal = {Quantum},
  volume = {5},
  pages = {581},
  year = {2021},
  doi = {10.22331/q-2021-11-16-581}
}

@article{Emerson2005,
  author = {Emerson, Joseph and Alicki, Robert and Zyczkowski, Karol},
  title = {Scalable noise estimation with random unitary operators},
  journal = {J. Opt. B},
  volume = {7},
  pages = {S347},
  year = {2005},
  doi = {10.1088/1464-4266/7/10/021}
}

@article{Emerson2007,
  author = {Emerson, Joseph and Silva, Marcus and Moussa, Osama and Ryan, Colm and Laforest, Martin and Baugh, Jonathan and Cory, David G. and Laflamme, Raymond},
  title = {Symmetrized characterization of noisy quantum processes},
  journal = {Science},
  volume = {317},
  pages = {1893--1896},
  year = {2007},
  doi = {10.1126/science.1145699}
}

@article{Gross2007,
  author = {Gross, David and Audenaert, Koenraad and Eisert, Jens},
  title = {Evenly distributed unitaries: On the structure of unitary designs},
  journal = {J. Math. Phys.},
  volume = {48},
  pages = {052104},
  year = {2007},
  doi = {10.1063/1.2716992}
}

@article{Dankert2009,
  author = {Dankert, Christoph and Cleve, Richard and Emerson, Joseph and Livine, Etera},
  title = {Exact and approximate unitary 2-designs and their application to fidelity estimation},
  journal = {Phys. Rev. A},
  volume = {80},
  pages = {012304},
  year = {2009},
  doi = {10.1103/PhysRevA.80.012304}
}

@article{BrownFawzi2015,
  author = {Brown, Winton and Fawzi, Omar},
  title = {Decoupling with random quantum circuits},
  journal = {Commun. Math. Phys.},
  volume = {340},
  pages = {867--900},
  year = {2015},
  doi = {10.1007/s00220-015-2470-1}
}

@article{Brandao2016,
  author = {Brandao, Fernando G. S. L. and Harrow, Aram W. and Horodecki, Michal},
  title = {Local random quantum circuits are approximate polynomial-designs},
  journal = {Commun. Math. Phys.},
  volume = {346},
  pages = {397--434},
  year = {2016},
  doi = {10.1007/s00220-016-2706-8}
}

@article{Oszmaniec2022,
  author = {Oszmaniec, Michal and Sawicki, Adam and Horodecki, Michal},
  title = {Epsilon-nets, unitary designs, and random quantum circuits},
  journal = {IEEE Trans. Inf. Theory},
  volume = {68},
  pages = {989--1015},
  year = {2022},
  doi = {10.1109/TIT.2021.3128110}
}

@misc{HunterJones2018,
  author = {Hunter-Jones, Nicholas},
  title = {Operator growth in random quantum circuits with symmetry},
  year = {2018},
  eprint = {1812.08219},
  archivePrefix = {arXiv},
  primaryClass = {cond-mat.stat-mech},
  doi = {10.48550/arXiv.1812.08219}
}

@article{Keyserlingk2018,
  author = {von Keyserlingk, C. W. and Rakovszky, Tibor and Pollmann, Frank and Sondhi, S. L.},
  title = {Operator hydrodynamics, OTOCs, and entanglement growth in systems without conservation laws},
  journal = {Phys. Rev. X},
  volume = {8},
  pages = {021013},
  year = {2018},
  doi = {10.1103/PhysRevX.8.021013}
}

@article{Nahum2018,
  author = {Nahum, Adam and Vijay, Sagar and Haah, Jeongwan},
  title = {Operator spreading in random unitary circuits},
  journal = {Phys. Rev. X},
  volume = {8},
  pages = {021014},
  year = {2018},
  doi = {10.1103/PhysRevX.8.021014}
}

@article{Wallman2016,
  author = {Wallman, Joel J. and Emerson, Joseph},
  title = {Noise tailoring for scalable quantum computation via randomized compiling},
  journal = {Phys. Rev. A},
  volume = {94},
  pages = {052325},
  year = {2016},
  doi = {10.1103/PhysRevA.94.052325}
}

@article{Hashim2021,
  author = {Hashim, Akel and Naik, Ravi K. and Morvan, Alexis and Ville, Jean-Loup and Mitchell, Bradley and Kreikebaum, John Mark and Davis, Marc and Smith, Ethan and Iancu, Costin and O'Brien, Kevin P. and others},
  title = {Randomized compiling for scalable quantum computing on a noisy superconducting quantum processor},
  journal = {Phys. Rev. X},
  volume = {11},
  pages = {041039},
  year = {2021},
  doi = {10.1103/PhysRevX.11.041039}
}

@article{Ware2021,
  author = {Ware, Matthew and Ribeill, Guilhem and Riste, Diego and Ryan, Colm A. and Johnson, Blake and da Silva, Marcus P.},
  title = {Experimental Pauli-frame randomization on a superconducting qubit},
  journal = {Phys. Rev. A},
  volume = {103},
  pages = {042604},
  year = {2021},
  doi = {10.1103/PhysRevA.103.042604}
}

@article{Knill2005,
  author = {Knill, Emanuel},
  title = {Quantum computing with realistically noisy devices},
  journal = {Nature},
  volume = {434},
  pages = {39--44},
  year = {2005},
  doi = {10.1038/nature03350}
}

@article{Nielsen2002,
  author = {Nielsen, Michael A.},
  title = {A simple formula for the average gate fidelity of a quantum dynamical operation},
  journal = {Phys. Lett. A},
  volume = {303},
  pages = {249--252},
  year = {2002},
  doi = {10.1016/S0375-9601(02)01272-0}
}

@article{Choi1975,
  author = {Choi, Man-Duen},
  title = {Completely positive linear maps on complex matrices},
  journal = {Linear Algebra Appl.},
  volume = {10},
  pages = {285--290},
  year = {1975},
  doi = {10.1016/0024-3795(75)90075-0}
}

@article{Jamiolkowski1972,
  author = {Jamiolkowski, Andrzej},
  title = {Linear transformations which preserve trace and positive semidefiniteness of operators},
  journal = {Rep. Math. Phys.},
  volume = {3},
  pages = {275--278},
  year = {1972},
  doi = {10.1016/0034-4877(72)90011-0}
}

@article{Choo2011Jury,
  author = {Choo, Younseok},
  title = {An Elementary Proof of the {Jury} Test for Real Polynomials},
  journal = {Automatica},
  volume = {47},
  number = {1},
  pages = {249--252},
  year = {2011},
  doi = {10.1016/j.automatica.2010.10.040}
}

@article{GonzalezDeMendoza2026,
  author        = {Gonzalez de Mendoza, Juan and Ostrove, Corey and Proctor, Timothy and Young, Kevin and Nielsen, Erik and Blume-Kohout, Robin},
  title         = {A Gauge-Invariant Theory of Small Markovian Errors in Quantum Gate Sets},
  journal       = {arXiv preprint arXiv:2608.14891},
  year          = {2026},
  eprint        = {2608.14891},
  archivePrefix = {arXiv},
  primaryClass  = {quant-ph}
}
\appendix

\section{Direct RB computation}
\label{app:drb-manual}

In this appendix we derive Proposition~\ref{prop:qdrb}. On the traceless Pauli sector, the ideal one-qubit $X_{\pi/2}$ and $Y_{\pi/2}$ gates act by the Bloch-sphere rotations
\begin{equation}
R_X=
\begin{pmatrix}
1&0&0\\
0&0&-1\\
0&1&0
\end{pmatrix},
\qquad
R_Y=
\begin{pmatrix}
0&0&1\\
0&1&0\\
-1&0&0
\end{pmatrix}.
\end{equation}
Write the Hamiltonian post-gate error vectors for the two native gates as
\begin{equation}
\mathbf h_X=(h_X^{(X)},h_Y^{(X)},h_Z^{(X)})^T,
\qquad
\mathbf h_Y=(h_X^{(Y)},h_Y^{(Y)},h_Z^{(Y)})^T,
\end{equation}
and let $[\mathbf h]_\times$ denote the $3\times 3$ cross-product matrix. On the traceless sector, a Hamiltonian post-gate error acts as $e^{A_g}$ with
\begin{equation}
A_X=2[\mathbf h_X]_\times,
\qquad
A_Y=2[\mathbf h_Y]_\times.
\end{equation}
Therefore the exact one-qubit $\Lmat$ matrix is
\begin{equation}
\Lmat(p)=p\,R_X\otimes e^{A_X}R_X + q\,R_Y\otimes e^{A_Y}R_Y.
\label{eq:drb-transfer-appendix}
\end{equation}
At the ideal point $A_X=A_Y=0$, the operator $\Lmat_0(p)$ has a simple unit eigenvalue with normalized eigenvector
\begin{equation}
    v=\frac{\operatorname{vec}(I_3)}{\sqrt 3}.
\end{equation}

\subsection{Ideal mixing}
\label{app:xy-mixing}
For completeness, the required design property can be checked directly. The rotations $R_X$ and $R_Y$ generate the one-qubit Clifford group, which is a unitary 2-design \cite{Dankert2009}. Let $x=pq$, $a=1-3x$, and $b=1-x$. The first-moment block has characteristic polynomial
\begin{equation}
    f_-(\lambda)=\lambda^3-\lambda^2+b\lambda-a,
\end{equation}
and the characteristic polynomial of $\Lmat_0(p)$ on the traceless second-moment sector factors as
\begin{equation}
    (\lambda-1)(\lambda^2-a)f_-(\lambda)f_+(\lambda),\qquad
    f_+(\lambda)=\lambda^3+\lambda^2+b\lambda+a.
\end{equation}
For $0<p<1$, $0<x\leq1/4$. The cubic Jury conditions \cite{Choo2011Jury} for $f_-$ reduce to
\begin{equation}
    a<1,\qquad 2x>0,\qquad 4(1-x)>0,\qquad x(4-9x)>0,
\end{equation}
so every root of $f_-$ lies strictly inside the unit disk. The roots of $f_+$ are their negatives, and the roots of $\lambda^2-a$ have modulus $\sqrt a<1$. Hence the Haar-fixed eigenvalue at one is simple and every non-Haar first- and second-moment eigenvalue has modulus below one. This proves the sequence-asymptotic 2-design property for each fixed $0<p<1$. At $p=0$ or $1$, unit-modulus non-Haar modes remain. Moreover,
\begin{align}
    \Delta_2(p)&=p+O(p^2) && (p\to0),\\
    \Delta_2(p)&=(1-p)+O((1-p)^2) && (p\to1).
    \label{eq:xy-gap-asymptotics}
\end{align}
so the convergence and perturbative neighborhood are not uniform near the endpoints.

For each ideal rotation, both $R_g\otimes R_g$ and its transpose fix
$v$. Thus $v$ is also the normalized left unit eigenvector of
$\Lmat_0(p)$, and $v^\perp$ is invariant. This is the one-qubit
instance of the orthogonal spectral-projector argument in
Sec.~\ref{sec:error-factorization}.

Let $P=\ket{v}\!\bra{v}$ and define the reduced resolvent
\begin{equation}
    S=\left((I-\Lmat_0)\big|_{v^\perp}\right)^{-1}(I-P),
\end{equation}
which is well defined for $0<p<1$ because the ideal gate set mixes on
the orthogonal complement of $v$.

Expanding Eq.~\eqref{eq:drb-transfer-appendix} gives
\begin{equation}
    \Lmat(p)=\Lmat_0+\Lmat_1+\Lmat_2+O(\norm{\mathbf h}^3),
\end{equation}
with
\begin{align}
    \Lmat_1
    &=p\,R_X\otimes A_XR_X
      +q\,R_Y\otimes A_YR_Y,\\
    \Lmat_2
    &=\frac{p}{2}\,R_X\otimes A_X^2R_X
      +\frac{q}{2}\,R_Y\otimes A_Y^2R_Y.
\end{align}
The first-order eigenvalue shift vanishes because
\begin{equation}
    \bra{v}\Lmat_1\ket{v}
    =\frac{p\,\tr A_X+q\,\tr A_Y}{3}=0,
\end{equation}
each $A_g$ being antisymmetric. Using the left eigenvector identified above, standard second-order perturbation theory for the simple eigenvalue at $1$ yields
\begin{equation}
    \gamma
    =
    1+\bra{v}\Lmat_2\ket{v}
    +\bra{v}\Lmat_1S\Lmat_1\ket{v}
    +O(\norm{\mathbf h}^3).
    \label{eq:drb-second-order-appendix}
\end{equation}
Evaluating Eq.~\eqref{eq:drb-second-order-appendix} and collecting coefficients gives the quadratic expansion
\begin{equation}
1-\gamma = \mathbf h^T H_{\mathrm{DRB}}(p)\mathbf h + O(\norm{\mathbf h}^3),
\end{equation}
where
\begin{equation}
\mathbf h=
\begin{pmatrix}
 h_X^{(X)} & h_Y^{(X)} & h_Z^{(X)} & h_X^{(Y)} & h_Y^{(Y)} & h_Z^{(Y)}
\end{pmatrix}^{\!T}
\end{equation}
and, writing $q=1-p$,
\begin{widetext}
\begin{equation}
H_{\mathrm{DRB}}(p)=\frac{4}{3}
\begin{pmatrix}
\displaystyle \frac{p(1+p^2)}{q} & p^2 & p^2 & p^2 & pq & -p^2\\[6pt]
p^2 & pq & pq & pq & q^2 & -pq\\
p^2 & pq & pq & pq & q^2 & -pq\\
p^2 & pq & pq & pq & q^2 & -pq\\
pq & q^2 & q^2 & q^2 & \displaystyle \frac{q(1+q^2)}{p} & -q^2\\[6pt]
-p^2 & -pq & -pq & -pq & -q^2 & pq
\end{pmatrix}.
\label{eq:drb-full-hessian}
\end{equation}
\end{widetext}
The three linearized gauge vectors generated by Eq.~\eqref{eq:lin-gauge} are
\begin{equation}
\begin{aligned}
    g_x &= (0,0,0,1,0,1)^T,\\
    g_y &= (0,1,-1,0,0,0)^T,\\
    g_z &= (0,1,1,-1,0,1)^T,
\end{aligned}
\end{equation}
and a direct check shows that
\begin{equation}
    H_{\mathrm{DRB}}(p)g_x=H_{\mathrm{DRB}}(p)g_y=H_{\mathrm{DRB}}(p)g_z=0.
\end{equation}
Together with positive definiteness of the reduced matrix in Eq.~\eqref{eq:Qdrb}, this shows that the full quadratic-form matrix has rank three and kernel equal to the gauge subspace.

To pass to \FOGI coordinates, choose the representatives
\begin{equation}
\begin{aligned}
f_1=(1,0,0,0,0,0)^T,
\\
f_2=(0,0,0,0,1,0)^T,
\\
f_3=(0,1,1,0,0,0)^T,
\end{aligned}
\end{equation}
which satisfy $F_i(f_j)=\delta_{ij}$ for the \FOGI coordinates of Eqs.~\eqref{eq:F1}--\eqref{eq:F3}. Because the gauge vectors lie in the kernel of the full quadratic-form matrix, the reduced quadratic form is independent of the chosen representatives, and its matrix is simply
\begin{equation}
\bigl(f_i^T H_{\mathrm{DRB}}(p) f_j\bigr)_{i,j=1}^3.
\end{equation}
Evaluating these matrix elements yields Eq.~\eqref{eq:Qdrb}.

\subsection{Eigenvalues of the reduced quadratic form}
\label{app:drb-spectrum}

Figure~\ref{fig:drb_eigs} shows the ordered eigenvalues of
$\Qmat_{\mathrm{DRB}}(p)$ on a logarithmic vertical scale.
As $p\to0$ or $1$, the largest coefficient diverges while the other
two vanish. This describes the nonuniform conditioning of the local
quadratic expansion as the ideal mixing gap closes; it should not
be interpreted as a divergent exact error rate.

\begin{figure}[t]
\centering
\includegraphics[width=0.76\linewidth]{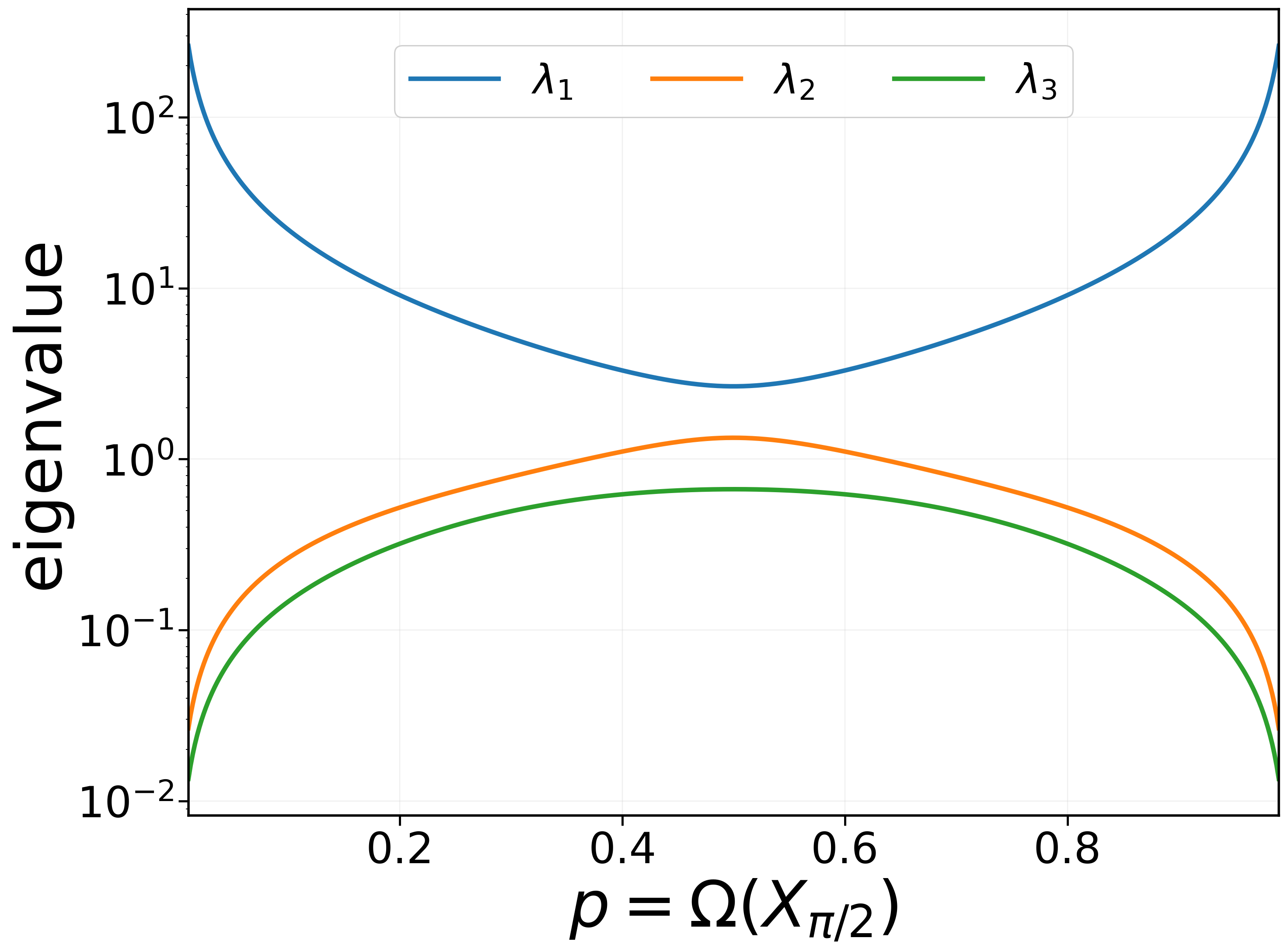}
\caption{Ordered eigenvalues of the reduced FOGI quadratic form $\Qmat_{\mathrm{DRB}}(p)$ in Eq.~\eqref{eq:Qdrb}, shown on a logarithmic vertical scale and truncated to $p\in[0.01,0.99]$. The largest coefficient diverges as the ideal moment gap closes at either endpoint, showing that the fixed-$p$ perturbative
expansion is nonuniform. The exact RB rate remains bounded, and the other two coefficients vanish because their modes require both gates.}
\label{fig:drb_eigs}
\end{figure}

\section{Direct perturbative derivation and symmetry interpretation of the CRB quadratic form}
\label{app:crb-manual}

In this appendix we derive Eq.~\eqref{eq:crb-quadratic-form} directly from the $\Lmat$ matrix. Using the notation of Sec.~\ref{sec:crb}, expand
\begin{equation}
\begin{aligned}
\mathscr L(\mathbf h)&=\mathscr L_0+\mathscr L_1+\mathscr L_2+O(\|\mathbf h\|^3),\\
\mathscr L_0&=\frac1{24}\sum_C R_C\otimes R_C,\\
\mathscr L_1&=\frac1{24}\sum_C R_C\otimes B_C R_C,\\
\mathscr L_2&=\frac1{48}\sum_C R_C\otimes B_C^2 R_C.
\end{aligned}
\end{equation}
Because the real three-dimensional representation of the one-qubit Clifford group is irreducible, the exact ideal twirl on the traceless Pauli sector is the rank-one projector
\begin{equation}
\mathscr L_0=P=\ket{v}\!\bra{v},
\qquad
v=\operatorname{vec}(I_3)/\sqrt3.
\end{equation}
Hence the reduced resolvent on the orthogonal complement is simply $I-P$. Since each $B_C$ is antisymmetric, $\bra{v}\mathscr L_1\ket{v}=0$, so second-order perturbation theory yields
\begin{equation}
\gamma(\mathbf h)
=
1+\bra{v}\mathscr L_2\ket{v}+\bra{v}\mathscr L_1(I-P)\mathscr L_1\ket{v}+O(\|\mathbf h\|^3).
\label{eq:crb-second-order-appendix}
\end{equation}

For the $\mathscr L_2$ term, use $(B\otimes C)\operatorname{vec}(X)=\operatorname{vec}(CXB^{\mathsf T})$ to obtain
\[
\bra{v}(R_C\otimes B_C^2R_C)\ket{v}=\frac13\tr(B_C^2).
\]
Therefore
\[
\bra{v}\mathscr L_2\ket{v}=\frac1{144}\sum_C \tr(B_C^2).
\]
Since $B_C=2[\mathbf h_C]_\times$ and $\tr([\mathbf x]_\times^2)=-2\|\mathbf x\|^2$, we get
\begin{equation}
\bra{v}\mathscr L_2\ket{v}=-\frac1{18}\sum_{C\in C_1}\|\mathbf h_C\|^2.
\label{eq:L2-crb-appendix}
\end{equation}

For the $\mathscr L_1(I-P)\mathscr L_1$ term, identifying matrices with their vectorizations, Eq.~\eqref{eq:u-w-main} implies
\[
\mathscr L_1(I_3)=\frac1{24}\sum_C B_C=\frac1{12}[\mathbf u(\mathbf h)]_\times,
\]
while the Frobenius adjoint satisfies
\[
\mathscr L_1^{\dagger}(I_3)
=-\frac1{24}\sum_C R_C^{\mathsf T}B_C R_C
=-\frac1{12}[\mathbf w(\mathbf h)]_\times.
\]
Both matrices are antisymmetric and therefore orthogonal to $I_3$, so $(I-P)$ acts trivially on them. Hence
\[
\bra{v}\mathscr L_1(I-P)\mathscr L_1\ket{v}
=\frac13\,\bigl\langle \mathscr L_1^{\dagger}(I_3),\mathscr L_1(I_3)\bigr\rangle_F.
\]
Using $\langle [\mathbf x]_\times,[\mathbf y]_\times\rangle_F=2\,\mathbf x\cdot\mathbf y$, we obtain
\begin{equation}
\bra{v}\mathscr L_1(I-P)\mathscr L_1\ket{v}
=-\frac1{216}\,\mathbf u(\mathbf h)\cdot\mathbf w(\mathbf h).
\label{eq:L1-crb-appendix}
\end{equation}

Combining Eqs.~\eqref{eq:crb-second-order-appendix}--\eqref{eq:L1-crb-appendix} gives
\begin{align}
1-\gamma(\mathbf h)
&=\frac1{18}\sum_{C\in C_1}\|\mathbf h_C\|^2
+\frac1{216}\,\mathbf u(\mathbf h)\cdot\mathbf w(\mathbf h)\\
&\quad+O(\|\mathbf h\|^3).
\end{align}
The displayed degree-two coefficient is $Q_{\mathrm{CRB}}^{(2)}$ in Eq.~\eqref{eq:crb-quadratic-form}.

\subsection{Symmetry interpretation}

With the metric of Eq.~\eqref{eq:crb-metric}, Eq.~\eqref{eq:crb-quadratic-form} organizes the coherent CRB quadratic form into three invariant subspaces. The orthogonal decomposition depends on this metric, while the gauge range and quotient are canonical \cite{GonzalezDeMendoza2026}. Define the Clifford action on $V=(\mathbb R^3)^{24}$ by
\[
(T_S\mathbf h)_C = R_S\,\mathbf h_{S^{-1}CS},
\qquad S\in C_1,
\]
and define
\[
(g_{\mathbf a})_C=\mathbf a-R_C\mathbf a,
\qquad
(d_{\mathbf a})_C=\mathbf a+R_C\mathbf a.
\]
Then
\begin{align}
V_{\rm gauge}&=\{g_{\mathbf a}:\mathbf a\in\mathbb R^3\},\\
V_{\rm dg}&=\{d_{\mathbf a}:\mathbf a\in\mathbb R^3\},\\
V_\perp&=\{\mathbf h: \mathbf u(\mathbf h)=0,\ \mathbf w(\mathbf h)=0\}.
\end{align}
Because $\sum_{C\in C_1}R_C=0$,
\begin{align}
\mathbf u(g_{\mathbf a})&=24\,\mathbf a, &
\mathbf w(g_{\mathbf a})&=-24\,\mathbf a,\\
\mathbf u(d_{\mathbf a})&=24\,\mathbf a, &
\mathbf w(d_{\mathbf a})&=24\,\mathbf a.
\end{align}
so $\dim V_{\rm gauge}=\dim V_{\rm dg}=3$. Moreover,
\[
\langle \mathbf h,d_{\mathbf a}\rangle
=
\mathbf a\cdot(\mathbf u(\mathbf h)+\mathbf w(\mathbf h)),
\qquad
\langle \mathbf h,g_{\mathbf a}\rangle
=
\mathbf a\cdot(\mathbf u(\mathbf h)-\mathbf w(\mathbf h)),
\]
which shows that $V_\perp$ is orthogonal to both $V_{\rm gauge}$ and $V_{\rm dg}$. Also,
\[
\langle d_{\mathbf a},g_{\mathbf b}\rangle
=
\mathbf b\cdot(\mathbf u(d_{\mathbf a})-\mathbf w(d_{\mathbf a}))=0,
\]
so the three subspaces are mutually orthogonal. Since $V=(\mathbb R^3)^{24}$ has dimension $72$, we obtain the orthogonal decomposition
\[
V=V_{\rm gauge}\oplus V_{\rm dg}\oplus V_\perp,
\qquad
\dim V_\perp=66.
\]

The gauge subspace is the tangent space to an exact gauge orbit. Indeed, for any $\mathbf a\in\mathbb R^3$, let $M(t)=e^{2t[\mathbf a]_\times}$ and define
\[
\widetilde R_C(t)=M(t)R_CM(t)^{-1}.
\]
This gate set is exactly gauge-equivalent to the ideal one for every $t$, so its CRB decay is identically ideal. Refactoring into post-gate form gives tangent vector $g_{\mathbf a}$ at $t=0$, and therefore the quadratic form vanishes on $V_{\rm gauge}$. Equivalently, substituting $g_{\mathbf a}$ into Eq.~\eqref{eq:crb-quadratic-form} gives
\[
\sum_C \|g_{\mathbf a,C}\|^2=48\|\mathbf a\|^2,
\qquad
\mathbf u(g_{\mathbf a})\cdot\mathbf w(g_{\mathbf a})=-576\|\mathbf a\|^2,
\]
so $Q_{\mathrm{CRB}}^{(2)}(g_{\mathbf a})=0$.

On $V_{\rm dg}$, the same identities give
\[
\sum_C \|d_{\mathbf a,C}\|^2=48\|\mathbf a\|^2,
\qquad
\mathbf u(d_{\mathbf a})\cdot\mathbf w(d_{\mathbf a})=576\|\mathbf a\|^2,
\]
and therefore
\[
Q_{\mathrm{CRB}}^{(2)}(d_{\mathbf a})
=
\frac{48}{18}\|\mathbf a\|^2+
\frac{576}{216}\|\mathbf a\|^2
=
\frac{16}{3}\|\mathbf a\|^2
=
\frac19\|d_{\mathbf a}\|^2.
\]
Finally, if $\mathbf h\in V_\perp$, then $\mathbf u(\mathbf h)=\mathbf w(\mathbf h)=0$, so
\[
Q_{\mathrm{CRB}}^{(2)}(\mathbf h)=\frac1{18}\|\mathbf h\|^2.
\]
Thus
\[
D_{\rm CRB}=\frac19\,\Pi_{\rm dg}+\frac1{18}\,\Pi_\perp,
\qquad
\ker D_{\rm CRB}=V_{\rm gauge},
\]
which is Eq.~\eqref{eq:Dcrb-main}.

As a normalization check, consider a gate-independent coherent post-gate error, $\mathbf h_C=\mathbf a$ for every $C\in C_1$. Since $\sum_C R_C=0$, we have
$\mathbf u(\mathbf h)=24\mathbf a$ and $\mathbf w(\mathbf h)=0$. Equation~\eqref{eq:crb-quadratic-form}
therefore gives
\begin{equation}
    Q_{\mathrm{CRB}}^{(2)}(\mathbf h)
    =\frac{24}{18}\|\mathbf a\|^2
    =\frac43\|\mathbf a\|^2.
\end{equation}
The quadratic prediction for the RB error rate is consequently $\frac34 Q_{\mathrm{CRB}}^{(2)}(\mathbf h)=\|\mathbf a\|^2$. This agrees with the leading term of the entanglement infidelity of the common coherent error,
\begin{equation}
    \sin^2\|\mathbf a\|
    =\|\mathbf a\|^2+O(\|\mathbf a\|^4),
\end{equation}
recovering the standard gate-independent RB result to the order considered here.
\end{document}